\documentclass[journal,twocolumn]{IEEEtran}
\usepackage[dvips]{graphicx} % for figures
\usepackage{amsfonts}
\usepackage{amssymb}
\usepackage{amscd}
\usepackage{amsmath}    % need for subequations
\usepackage{bbm}
\usepackage{mathabx}
\usepackage{enumerate}
\usepackage{epsfig}
\usepackage{subfigure}
\usepackage{xcolor}
\usepackage{amsthm}
\usepackage{diagbox}
\usepackage{amsthm}
\usepackage{mathdots}
\usepackage{bm}
\usepackage{nicefrac}
\usepackage{tabularx}
\usepackage{multirow}
\usepackage[justification=justified,font=footnotesize]{caption}

\usepackage{algcompatible}
\usepackage{newfloat}

\newtheorem{theorem}{Theorem}
\newtheorem{lemma}[theorem]{Lemma}
\newtheorem{corollary}[theorem]{Corollary}

\newtheorem{definition}[theorem]{Definition}

\newtheorem{remark}[theorem]{Remark}
\newtheorem{proposition}[theorem]{Proposition}

\newlength{\blank}
\renewenvironment{proof}[1][{\hspace{-\blank}}]{{\noindent\textbf{Proof~{#1}.\ }}}{\hfill\qed}
\newcommand{\ket}[1]{|#1\rangle}
\newcommand{\bra}[1]{\langle#1|}
\newcommand{\braket}[2]{\langle #1 | #2 \rangle}
\newcommand{\ketbra}[2]{| #1 \rangle\langle #2 |}
\newcommand{\proj}[1]{| #1 \rangle\langle #1 |}
\mathchardef\ordinarycolon\mathcode`\:
\mathcode`\:=\string"8000
\def\vcentcolon{\mathrel{\mathop\ordinarycolon}}
\begingroup \catcode`\:=\active
  \lowercase{\endgroup
  \let :\vcentcolon
  }

\newcommand{\nc}{\newcommand}
\nc{\rnc}{\renewcommand}
\nc{\beq}{\begin{equation}}
\nc{\eeq}{{\end{equation}}}
\nc{\beqa}{\begin{eqnarray}}
\nc{\eeqa}{\end{eqnarray}}
\nc{\lbar}[1]{\overline{#1}}
\nc{\avg}[1]{\langle#1\rangle}
\nc{\Rank}{\operatorname{Rank}}
\nc{\smfrac}[2]{\mbox{$\frac{#1}{#2}$}}
\nc{\tr}{\operatorname{Tr}}
\newcommand{\Tr}{\operatorname{Tr}}
\nc{\ox}{\otimes}
\nc{\dn}{\downarrow}
\nc{\cA}{\mathcal{A}}
\nc{\cB}{\mathcal{B}}
\nc{\cC}{\mathcal{C}}
\nc{\cD}{\mathcal{D}}
\nc{\cE}{\mathcal{E}}
\nc{\cF}{\mathcal{F}}
\nc{\cG}{\mathcal{G}}
\nc{\cH}{\mathcal{H}}
\nc{\cI}{\mathcal{I}}
\nc{\cJ}{\mathcal{J}}
\nc{\cK}{\mathcal{K}}
\nc{\cL}{\mathcal{L}}
\nc{\cM}{\mathcal{M}}
\nc{\cN}{\mathcal{N}}
\nc{\cO}{\mathcal{O}}
\nc{\cP}{\mathcal{P}}
\nc{\cR}{\mathcal{R}}
\nc{\cS}{\mathcal{S}}
\nc{\cT}{\mathcal{T}}
\nc{\cX}{\mathcal{X}}
\nc{\cZ}{\mathcal{Z}}
\nc{\csupp}{{\operatorname{csupp}}}
\nc{\qsupp}{{\operatorname{qsupp}}}
\nc{\var}{\operatorname{var}}
\nc{\rar}{\rightarrow}
\nc{\lrar}{\longrightarrow}
\nc{\polylog}{\operatorname{polylog}}
\nc{\1}{\mathbbm{1}}
\nc{\id}{{\operatorname{id}}}
\nc{\RR}{{{\mathbb R}}}
\nc{\CC}{{{\mathbb C}}}
\nc{\FF}{{{\mathbb F}}}
\nc{\NN}{{{\mathbb N}}}
\nc{\ZZ}{{{\mathbb Z}}}
\nc{\PP}{{{\mathbb P}}}
\nc{\QQ}{{{\mathbb Q}}}
\nc{\UU}{{{\mathbb U}}}
\nc{\EE}{{{\mathbb E}}}

\nc{\GOCC}{{\mathrm{GOCC}}}
\nc{\Wplus}{{\mathrm{W+}}}

\newcommand{\ip}[2]{\langle #1|#2\rangle}
\newcommand{\op}[2]{|#1\rangle \langle #2|}

\newcommand{\mbf}{\mathbf}

\newcommand{\mc}{\mathcal}

\newcommand{\mf}[1]{\mathfrak{ #1}}

\renewenvironment{proof}{\medskip\noindent\textbf{Proof.}}{\hfill$\blacksquare$\medskip}
\newenvironment{proof-of}[1]{\medskip\noindent\textbf{Proof of {#1}.}}{\hfill$\blacksquare$\medskip}

\newcommand{\yunchao}[1]{{#1}}
\newcommand{\Qi}[1]{{#1}}
\newcommand{\andreas}[1]{{#1}}

\begin{document}

\title{One-Shot Coherence Distillation: \andreas{Towards Completing the Picture}}

\author{Qi Zhao,
        Yunchao Liu, Xiao Yuan, Eric Chitambar
        and~Andreas~Winter% <-this % stops a space
\thanks{Qi Zhao and Yunchao Liu are with Center for Quantum Information, Institute for Interdisciplinary Information Sciences,
Tsinghua University, Beijing 100084, China. Email: zhaoqithu10@gmail.com}% <-this % stops a space
\thanks{Xiao Yuan is with Department of Materials, University of Oxford, Parks Road, Oxford OX1 3PH, United Kingdom.}% <-this % stops a space
\thanks{Eric Chitambar is with the Department of Electrical and Computer Engineering at the University of Illinois, Urbana-Champaign, Illinois 61801, USA. Email: echitamb@illinois.edu.}%
\thanks{Andreas Winter is with ICREA---Instituci\'o Catalana de Recerca i Estudis Avan\c{c}ats, Pg.~Lluis Companys, 23, ES-08001 Barcelona, Spain and also with F\'{\i}sica Te\`{o}rica: Informaci\'{o} i Fen\`{o}mens Qu\`{a}ntics, Departament de F\'{\i}sica, Universitat Aut\`{o}noma de Barcelona, ES-08193 Bellaterra (Barcelona), Spain. Email: andreas.winter@uab.cat.}%
\thanks{24 January 2019.}}

%\author{Qi Zhao}
%\affiliation{Center for Quantum Information, Institute for Interdisciplinary Information Sciences,
%Tsinghua University, Beijing 100084, China}
%\email{zhaoqi14@mails.tsinghua.edu.cn}
%
%\author{Yunchao Liu}
%\affiliation{Center for Quantum Information, Institute for Interdisciplinary Information Sciences,
%Tsinghua University, Beijing 100084, China}
%
%\author{Xiao Yuan}
%\affiliation{Department of Materials, University of Oxford, Parks Road, Oxford OX1 3PH, United Kingdom}
%
%\author{Eric Chitambar}
%\affiliation{Department of Physics and Astronomy, Southern Illinois University,
%Carbondale, Illinois 62901, USA}
%\email{echitamb@siu.edu}
%
%\author{Andreas Winter}
%\affiliation{ICREA---Instituci\'o Catalana de Recerca i Estudis Avan\c{c}ats, %
%Pg.~Lluis Companys, 23, ES-08001 Barcelona, Spain}
%\affiliation{F\'{\i}sica Te\`{o}rica: Informaci\'{o} i Fen\`{o}mens Qu\`{a}ntics,
%Departament de F\'{\i}sica, Universitat Aut\`{o}noma de Barcelona, ES-08193 Bellaterra (Barcelona), Spain}
%\email{andreas.winter@uab.cat}

\maketitle

\begin{abstract}
The resource framework of quantum coherence was introduced
by Baumgratz, Cramer and Plenio [\emph{Phys. Rev. Lett.} 113, 140401 (2014)]
and further developed by Winter and Yang [\emph{Phys. Rev. Lett.} 116, 120404 (2016)].
We consider the one-shot problem of distilling pure coherence from
a single instance of a given resource state. Specifically, we determine
the distillable coherence with a given fidelity under
incoherent operations (IO) through a generalisation of the Winter-Yang
protocol.  This is compared to the distillable coherence under
maximal incoherent operations (MIO) and
dephasing-covariant incoherent operations (DIO), which can be cast as a
semidefinite programme, that has been presented previously by Regula
\emph{et al.} [\emph{Phys. Rev. Lett.} 121, 010401 (2018)]. Our results are given in terms of
a smoothed min-relative entropy distance from the incoherent set of states,
and a variant of the hypothesis-testing relative entropy distance, respectively.
The one-shot distillable coherence is also related to one-shot randomness extraction.
Moreover, from the one-shot formulas under IO, MIO, DIO, we can recover the optimal distillable
rate in the many-copy asymptotics, yielding the relative entropy of
coherence. These results can be compared with previous work by some of the present
authors [Zhao \emph{et al.}, \emph{Phys. Rev. Lett.} 120, 070403 (2018)]
on one-shot coherence formation under IO,
MIO, DIO and also SIO. This shows that the amount
of distillable coherence is essentially the same for IO, DIO, and MIO,
despite the fact that the three classes of operations are very different.
We also relate the distillable coherence under strictly incoherent
operations (SIO) to a constrained hypothesis testing problem and explicitly
show the existence of bound coherence under SIO in the asymptotic regime.
\end{abstract}

%\date{8 June 2018}

\section{Introduction}
\label{sec:intro}
Coherence, as the signature of non-classicality, has  applications in many quantum information processing tasks, including cryptography \cite{coles2016numerical}, metrology \cite{giovannetti2011advances}, randomness generation \cite{Yuan15intrinsic, ma2017source}, biological systems \cite{plenio2008dephasing, rebentrost2009role} and thermodynamics \cite{Aberg14,Horodecki15,Lostaglio15,lostaglio2015description,narasimhachar2015low}.
Recently, a resource theory framework of quantum coherence has been introduced by Baumgratz, Cramer and Plenio \cite{baumgratz14}
(after prior work of \AA{}berg \cite{Aaberg2006}, as well as Braun and Georgeot \cite{BraunGeorgeot}).  A general quantum resource theory consists of two objects: a set of free states and a set of free operations, with the latter acting invariantly on the former.  In the resource theory of coherence, the free states $\mathcal{I}$ for a $d$-dimensional Hilbert space %$\mathcal{H}_d$
is the set of density matrices that are invariant under conjugation by the phase unitary
\begin{equation}\label{phaseunitary}
Z = \left(\begin{array}{cccc}
        1 & 0            & \cdots & 0 \\
        0 & e^{2\pi i/d} &        & \vdots \\
        \vdots &         & \ddots & 0  \\
        0 & \cdots       &      0 & e^{2\pi i(d-1)/d}
      \end{array}\right),
\end{equation}
with $Z$ being expressed in \textit{a priori} fixed computational basis $\{\ket{x}\}_{x=1}^{d}$.
Furthermore, different physical and mathematical motivations have led to the proposal and study of different classes of free operations, most notably the maximally incoherent operations (MIO) \cite{Aaberg2006}, the dephasing-covariant incoherent operation (DIO) \cite{Chitambar16prl,Marvian16}, the incoherent operations (IO) \cite{baumgratz14}, and the strictly incoherent operations (SIO) \cite{Winter16,Yadin16}.  For instance, SIO emerges as a natural class of operations to consider when quantifying visibility in interferometer experiments \cite{Biswas-2017a}.  Several coherence measures have been introduced to quantify the amount of coherence in a state, including the relative entropy and $\ell_1$-norm of coherence \cite{baumgratz14}, the coherence of formation \cite{Yuan15intrinsic,Aaberg2006}, and the robustness of coherence \cite{Napoli16}, which have all been shown to possess different operational meanings \cite{Yuan15intrinsic,Winter16,Napoli16,Rana17}. An overview on  recent developments of the resource theory of coherence can be found in Ref.~\cite{Streltsov-review}.

An important problem in the resource theory is convertibility of states via free operations, especially those between an arbitrary state $\rho$ and a maximal resource state $\ket{\Psi_M}$.  In an $M$-level system, the superposition state $\ket{\Psi_M}=\frac{1}{\sqrt{M}}\sum_{i=1}^{M}\ket{i}$ has the maximal coherence with respect to the aforementioned, and indeed all, coherence measures, and the qubit state $\ket{\Psi_2}$ serves as resource unit in the theory, the ``cosbit'' \cite{Diaz18}.
The process of converting a given state $\rho$ to $\ket{\Psi_M}$ is referred to as \emph{coherence distillation} and the reverse process as \emph{coherence dilution}.
In the asymptotic case where an unbounded number of independent and identically distributed
(i.i.d.) copies of the initial state are provided, the optimal rates of asymptotic distillation and dilution are quantified by the relative entropy of coherence and the coherence of formation, respectively \cite{Winter16}.

In practical scenarios, i.i.d. resources are not available and an analysis of non-asymptotic tasks becomes crucial. In recent work \cite{Zhao-one-shot}, the formation problem was addressed in the one-shot setting under all four of the above operational classes.  Subsequently, the converse question of one-shot coherence distillation was solved for the classes MIO and DIO \cite{Regula-one-shot}.  In the present paper, we close the remaining gap by providing an analysis of general one-shot coherence distillation under IO and SIO. Our results show that, while not being exactly identical, the three classes MIO, DIO, and IO all lead to essentially the same expression for the distillable coherence in terms of a min-relative entropy distance from the incoherent states; the differences are in the smoothing parameters and universal additive terms. When extended to the asymptotic case, our one-shot results imply that MIO, DIO, and IO all have distillable coherence given by the relative entropy, thereby recovering earlier work found in \cite{Winter16}.  In contrast, we show that the coherence distillation by SIO can behave quite differently than its more general counterparts, and we characterize the one-shot distillable coherence as a constrained hypothesis testing problem. We further show that coherence can be a \textit{bound} resource under SIO, meaning that states exist with zero distillable coherence but nonzero coherence cost.
We also show the relationship among the extractable randomness, IO and DIO distillable rate, which are are indeed closely related. The reason is that the distillation proceeds by a kind of decoupling process, and can thus be related to randomness extraction.

The structure of the remainder of the paper is as follows:
In Section~\ref{sec:dilution+distillation} we review the four mentioned
classes of incoherent operations and introduce the dilution and distillation
tasks for coherence in the one-shot setting. We present the prior results
on dilution in all four cases, and in Section~\ref{sec:dios-mio} review
the prior results on distillation under MIO and DIO. In Section~\ref{sec:des-pudels-kern:io}
we come to the first main result of the present paper, a tight one-shot
characterization of distillation with IO, showing in particular a distillation
protocol achieving the lower bound with an operation that is at the same time
IO and DIO.
Then, in Section~\ref{sec:sio} we give a one-shot formula for distillation under SIO,
and most importantly, show the existence of bound coherence in this model.
\yunchao{In Section~\ref{sec:AEP} we show how the obtained
one-shot formulas for IO, DIO and MIO imply the previously known asymptotic
distillation rate $C_r(\rho)$, and finally Section~\ref{sec:randomness} contains the discussion on the relation between randomness
extraction and coherence distillation, after which we conclude.}

\section{Coherence Dilution \protect\\ and Distillation}
\label{sec:dilution+distillation}
%We begin by describing the four classes of incoherent operations typically considered in the resource theory of coherence.  We then introduce the problems of coherence dilution and distillation. The latter is studied in detail in Sections \ref{sec:dios-mio}--\ref{sec:sio}.

\subsection{Classes of incoherent operations}

A general resource theory of coherence is constructed by imposing restrictions on the allowed completely positive trace-preserving (CPTP) maps $\Lambda:\mc{L}(A)\to\mc{L}(B)$, where $\mc{L}(A)$ denotes the space of linear operators acting on a finite-dimensional Hilbert space $A$ and likewise for $\mc{L}(B)$.  One necessary restriction is that $\Lambda$ acts invariantly on the set of incoherent states $\mc{I}$; i.e. $\Lambda(\delta)\in\mc{I}$ whenever $\delta\in\mc{I}$.  The completely dephasing map $\Delta(\rho)=\sum_{x=1}^{d}\op{x}{x}\rho\op{x}{x}$ plays an important role in this theory as it destroys all coherence in a state, and it therefore maps any resource state to a free one.  The following operations classes have been proposed in \cite{baumgratz14,Aaberg2006,Chitambar16prl,Marvian16,Winter16,Yadin16}, each motivated by different physical considerations.

\begin{LaTeXdescription}
  \item[MIO] \emph{Maximal Incoherent Operations} \cite{Aaberg2006} are characterized by
             $\Lambda(\delta)\in\mathcal{I}$ for all $\delta\in\mathcal{I}$, which
             may be expressed as the identity $\Lambda \circ \Delta = \Delta \circ \Lambda \circ \Delta$;
  \item[DIO] \emph{Dephasing-Covariant Incoherent Operations} \cite{Chitambar16prl,Marvian16}
             satisfy the stronger condition
             $\Lambda \circ \Delta = \Delta \circ \Lambda$;
  \item[IO]  \emph{Incoherent Operations} \cite{baumgratz14} are those \Qi{MIO} with a Kraus decomposition
             $\Lambda(\rho) =\sum_\alpha K_\alpha\rho K_\alpha^\dagger$,
             such that $K_\alpha\delta K_\alpha^\dagger/\Tr(K_\alpha\delta K_\alpha^\dagger)\in\mathcal{I}$
             for all $\alpha$ and $\delta\in\mathcal{I}$;
  \item[SIO] \emph{Strictly Incoherent Operations} \cite{Winter16} are those IO with a
             Kraus decomposition
             $\Lambda(\rho)=\sum_\alpha K_\alpha\rho K_\alpha^\dag$
             such that $K_\alpha=\sum_{x} c_{\alpha,x}\ketbra{f_\alpha(x)}{x}$,
             and $f_\alpha:[d]=\{1,\cdots,d\}\to[d]$ is a one-to-one function for all $\alpha$.
\end{LaTeXdescription}

\noindent
The relationships among these classes are
\begin{equation}
\begin{aligned}
&\text{SIO} \varsubsetneq  \text{IO} \varsubsetneq \text{MIO}, \\
&\text{SIO} \varsubsetneq  \text{DIO} \varsubsetneq \text{MIO};
\end{aligned}
\end{equation}
note that IO is not included in DIO and vice versa.

For later use, and because it will turn out to be significant, we also
give a name to the intersection of IO and DIO,
\[
  \text{DIIO} = \text{IO} \cap \text{DIO},
\]
which we dub \emph{dephasing-covariant incoherent IO}.

%%%%%%%%%%%%%%%%%%%%%%%%%%%%%%%%%%%%%%%%%%%%%%%%%%%%%%%%%%%%%%%%%%%%%%%%%%%%%%%%%%%%%%%%%%
\begin{table*}[ht]
\centering
\begin{minipage}{0.81\textwidth}
\centering
\begin{tabular}{|rl|llll|}
  \hline
  & & \ MIO\phantom{======} & DIO \phantom{======} & IO & SIO\\[1mm]
  \hline
  \multirow{2}{*}{\phantom{=}One-shot $\biggl\{$}
    & distillation\phantom{=} & \ $\widetilde{C}_{H}^{\varepsilon}$ \cite{Regula-one-shot}
                        & $\widetilde{C}_{H}^{\varepsilon}$ \cite{Regula-one-shot}
                             & $\mathbf{ C_{\min}^{\varepsilon'}}$ \textbf{[$\ast$]} & $\mathbf{Thm.~\ref{th:SIO}}$ \textbf{[$\ast$]} \\[1mm]
    & formation    & \ $C_{\max}^{\varepsilon}$ \cite{Zhao-one-shot}
                        & $C_{\Delta,\max}^{\varepsilon}$ \cite{Zhao-one-shot}
                             & $C_{0}^{\varepsilon}$  \cite{Zhao-one-shot}
                                 & $C_{0}^{\varepsilon}$ \cite{Zhao-one-shot}\phantom{=} \\[1mm]
  \hline
  \multirow{2}{*}{\phantom{=}Asymptotic $\biggl\{$}
    & distillation & \ $C_r$ \cite{Regula-one-shot}
                        & $C_r$ \cite{Regula-one-shot}
                            & $C_r$ \cite{Winter16} &  $\mathbf{\neq  C_r}$ \textbf{[$\ast$]}   \\[1mm]
    & formation    & \ $C_r$ \cite{Zhao-one-shot,Eric:MIO-CoF}
                        & $C_r$ \cite{Zhao-one-shot,Eric:MIO-CoF}
                            & $C_f$ \cite{Winter16} & $C_f$ \cite{Winter16}\phantom{=} \\[1mm]
  \hline
\end{tabular}
\vspace{3mm}
\caption{\rm Coherence distillation and formation rates, with attributions where
         the corresponding result was first obtained. Our new contributions
         (in bold, and marked by asterisks)
         are the IO and SIO distillation rates in the one-shot setting;
         furthermore, we prove that MIO, DIO and IO one-shot
         distillation rates differ only by varying the smoothing parameter
         and by universal additive terms. The one-shot
         distillation rate under SIO is converted into an optimization problem in
         Theorem~\ref{th:SIO}, and in the asymptotic case, we prove the existence
         of bound coherence under SIO, showing that the rate is different from
         the relative entropy of coherence in general.}
\label{Table:results}
\end{minipage}
\end{table*}
%%%%%%%%%%%%%%%%%%%%%%%%%%%%%%%%%%%%%%%%%%%%%%%%%%%%%%%%%%%%%%%%%%%%%%%%%%%%%%%%%%%%%%%%%%

\subsection{Coherence dilution}

The one-shot coherence dilution problem characterizes the minimal resource required for the formation of a target state with an allowed error $\varepsilon$. The definition of one-shot coherence dilution is as follows.
\begin{definition}
  Let $\mathcal{O}\in\{\text{MIO},\,\text{DIO},\,\text{IO},\,\text{SIO}\}$ denote some class of
  incoherent operations. Then for a given state $\rho$ and $\varepsilon\ge 0$,
  the one-shot coherence formation cost with error $\varepsilon$ under $\mathcal{O}$
  is defined as
  \begin{equation}\label{Eq:formationIO}
    C_{c,\mathcal{O}}^\varepsilon(\rho)
        = \min_{\Lambda\in \mathcal{O}}
          \bigl\{\log M \,:\, F(\Lambda(\Psi_M),\rho)^2 \geq 1-\varepsilon \bigr\}.
  \end{equation}
  where $F(\rho,\sigma) = \Tr \sqrt{\sqrt{\rho}\sigma\sqrt{\rho}}$
  is the usual mixed state fidelity.
%  and $\ket{\Psi_M}=\frac{1}{\sqrt{M}}\sum_{i=1}^{M}\ket{i}$
%  is the maximally coherent state of rank $M$.
  By default, here and throughout the paper, $\log \equiv \log_2$ is the binary
  logarithm, in accordance with information theoretic use.
\end{definition}

%\noindent
In Ref. \cite{Zhao-one-shot}, several coherence monotones were proposed to estimate
the one-shot coherence cost.  The first two are based on the max relative entropy
$D_{\max}(\rho\|\sigma)=\log\min\{\lambda\;|\;\rho\leq\lambda\sigma\}$. \yunchao{Alternatively, one can introduce this quantity using the \andreas{sandwiched quantum $\alpha$-R\'enyi} divergence
\begin{equation}\label{alphadivergence}
    \tilde{D}_{\alpha}(\rho\|\sigma)=\frac{1}{\alpha-1}\log\left(\Tr\left[\left(\sigma^{\frac{1-\alpha}{2\alpha}}\rho\sigma^{\frac{1-\alpha}{2\alpha}}\right)^{\alpha}\right]\right),
  \end{equation}
letting $\alpha\to\infty$~\cite{2013MartinRenyi}.}
One then defines
\begin{align}
C_{\max}(\rho)&=\min_{\delta\in\mc{I}}D_{\max}(\rho\|\delta)=\log\min\{\lambda\;|\;\rho\leq\lambda\delta\},\\
C_{\Delta,\max}(\rho)&=\min_{\sigma\in \overline{A_\rho}}D_{\max}(\rho\|\sigma)=\log\min\{\lambda\;|\;\rho\leq\lambda\Delta(\rho)\},
\end{align}
where $A_\rho=\{\sigma\geq 0\;|\;\sigma+t\rho=(1+t)\Delta(\rho),\; t>0\}$ \cite{Chitambar16pra} and $\overline{A_\rho}$ is its closure.  The quantity $C_{\max}(\rho)$ is a monotone under MIO, while $C_{\Delta,\max}(\rho)$ is a monotone under DIO, but not IO in general.  However, for a pure state $\ket{\psi}$, $C_{\Delta,\max}$ reduces to its \textit{incoherent rank} \cite{Chitambar16pra}, defined as
\begin{equation}
C_0(\psi):=S_0(\Delta(\psi))=\log\text{rank}[\Delta(\psi)],
\end{equation}
and the incoherent rank is a monotone under IO even for stochastic pure state transformations \cite{baumgratz14}.  Thus, one can obtain a general mixed state monotone for IO by taking a convex roof extension:
\begin{align}
C_{0}(\rho)&:=\min_{p_i,\ket{\psi_i}}\max_{i} C_{0}(\psi_i),
\end{align}
where the minimization is over all pure state ensembles satisfying $\rho=\sum_i p_i\op{\psi_i}{\psi_i}$.  Note, the minimization is followed by a maximization rather than an averaging over the $\ket{\psi_i}$ (as typically done in such measures) because $C_{0}(\psi)$ is a stochastic IO monotone.

Since we allow $\varepsilon$ error in our one-shot tasks, we likewise apply an $\varepsilon$ smoothing to our coherence measures.  Depending on the quantity, this is done by either minimizing or maximizing the measure over all states $\rho'$ lying in an $\varepsilon$-ball around $\rho$.  For example, if $C$ represents any one of the measures $C_{\max}$, $C_{\Delta,\max}$ or $C_0$, its $\varepsilon$-smoothed variant is given by
\begin{equation}
C^\varepsilon(\rho)=\min_{\rho'\in B_\varepsilon(\rho)}C(\rho').
\end{equation}
Here $B_\varepsilon(\rho)=\{\rho':P(\rho',\rho)<\varepsilon\}$ is
defined with respect to the purified distance $P(\rho',\rho) = \sqrt{1-F(\rho',\rho)^2}$.
Note that the purified distance is related to the trace distance by \cite{Fuchs-vandeGraaf}
\[
  1-\sqrt{1-P(\rho',\rho)^2}\leq \frac12\|\rho-\rho'\|_1 \leq P(\rho',\rho).
\]
As shown in \cite{Zhao-one-shot}, the measures $C^\varepsilon_{\max}$, $C^\varepsilon_{\Delta,\max}$
and $C^\varepsilon_0$ precisely characterize the one-shot coherence formation for MIO, DIO,
and IO/SIO, respectively (see Table \ref{Table:results}). Note that our definition of
$B_\varepsilon(\rho)$
differs from \cite{Zhao-one-shot} by replacing $\varepsilon \leftrightarrow \sqrt{\varepsilon}$.

The \yunchao{regularized} (many-copy) coherence cost of formation can then be defined using the one-shot quantities.
For $\mathcal{O}\in\{\text{MIO},\,\text{DIO},\,\text{IO},\,\text{SIO}\}$, the asymptotic rate of coherence cost for a state $\rho$ under operations $\mc{O}$ is defined as
\yunchao{
\begin{equation}\label{asymptoticformation}
  C_{c,\mathcal{O}}^{\infty}(\rho) := \limsup_{\varepsilon \to 0^+}\limsup_{n\to\infty}\frac{1}{n}C_{c,\mathcal{O}}^\varepsilon(\rho^{\otimes n}).
\end{equation}
}
Remarkably, for all four operational classes, the coherence cost has a single-letter characterization.
For MIO and DIO, the \yunchao{regularized coherence cost} is given by the relative entropy of coherence, which we call the asymptotic coherence cost under $\text{MIO},\,\text{DIO}$
\cite{Zhao-one-shot,Eric:MIO-CoF},
\begin{equation}
  C_r(\rho)=\min_{\delta\in\mc{I}}D(\rho\|\delta)=S(\Delta(\rho))-S(\rho),
\end{equation}
where $D(X\|Y)=-\tr[X(\log Y-\log X)]$ is the relative entropy\yunchao{, and we have
\begin{equation}\label{MIODIOdilution}
    \lim_{\varepsilon\to 0^+}\lim_{n\to\infty}\frac{1}{n}C_{c,MIO/DIO}^\varepsilon\left(\rho^{\otimes n}\right)=C_r(\rho).
\end{equation}
In the following, we use $\lim$ instead of $\limsup$ or $\liminf$ for simplicity.} On the other hand, for IO and SIO, the \yunchao{regularized coherence cost
is given by the so-called coherence of formation \cite{Winter16}, which we call the asymptotic coherence cost under $\text{IO},\,\text{SIO}$ and is defined as
\begin{equation}
\label{Eq:CoF}
C_f(\rho)=\min_{p_i,\ket{\psi_i}}\sum_i p_i C_{r}(\psi_i),
\end{equation}
where
\begin{equation}
    \lim_{\varepsilon\to 0^+}\lim_{n\to\infty}\frac{1}{n}C_{c,IO/SIO}^\varepsilon\left(\rho^{\otimes n}\right)=C_f(\rho).
\end{equation}
}
In contrast to the entanglement of formation, the coherence of formation is an
additive function \cite{Winter16}.

\subsection{Coherence distillation}
We next turn to the main focus of this paper, which is the task of transforming a given state $\rho$ into a maximally coherent pure state.  The one-shot distillation of the problem with $\varepsilon$ error is stated as follows.
\begin{definition}
  Let $\mathcal{O}\in\{\text{MIO},\,\text{DIO},\,\text{IO},\,\text{SIO}\}$ denote some class of
  incoherent operations. Then for a given state $\rho$ and $\varepsilon\ge 0$,
  the one-shot coherence distillation rate with error $\varepsilon$ under $\mathcal{O}$
  is defined as
  \begin{equation}\label{Eq:distillationIO}
    C_{d,\mathcal{O}}^\varepsilon(\rho)
        = \max_{\Lambda\in \mathcal{O}}
          \bigl\{\log M \,:\, F(\Lambda(\rho),\Psi_M)^2 \geq 1-\varepsilon \bigr\}.
  \end{equation}
%  where $F(\rho,\sigma) = \Tr \sqrt{\sqrt{\rho}\sigma\sqrt{\rho}}$
%  is the usual mixed state fidelity, and
%  $\ket{\Psi_M}=\frac{1}{\sqrt{M}}\sum_{i=0}^{M-1}\ket{i}$
%  is the maximally coherent state of rank $M$.
\end{definition}

%\noindent
In order to quantify the one-shot distillable rates, we first introduce two additional functions.  The first one is based on the min quantum R\'{e}nyi divergence \yunchao{$D_{\min}(\rho\|\sigma)=-\log F(\rho,\sigma)^2=\tilde{D}_{1/2}(\rho\|\sigma)$~\cite{2013MartinRenyi}, $\alpha=\frac{1}{2}$ in Eq.~\eqref{alphadivergence}}.  Using this quantity, the \textit{min-entropy of coherence} $C_{\min}(\rho)$ can be defined as
\begin{equation}
\label{cmin}
\begin{aligned}
	C_{\min}(\rho) &= \min_{\delta\in\mathcal{I}}D_{\min}(\rho\|\delta).
\end{aligned}
\end{equation}
The properties of $C_{\min}$ and its relationship between other coherence monotones are discussed in \cite{Liu18}.

The second quantifier of coherence relevant to our study is based on the hypothesis testing relative entropy \cite{Wang12,Dupuis-entropies}. In the task of hypothesis testing, a positive
operator valued measure (POVM) with two elements $\{W,\1-W\}$ is used to distinguish
two possible states $\rho$ and $\sigma$.
The probability of obtaining a correct guess on $\rho$ is $\mathrm{Tr}(\rho W)$ and the probability of obtaining a wrong guess on $\sigma$ is $\mathrm{Tr}(\sigma W)$. The hypothesis testing relative entropy characterizes the minimal wrong guessing probability on $\sigma$, with the constraint that the probability of the correct guess on $\rho$ is no less than $1-\varepsilon$:
\begin{equation}\begin{split}
  \label{dh}
  D_H^\varepsilon(\rho\|\sigma)
     = -\log \min \bigl\{\Tr \sigma \Qi{W} \,:\, 0 \leq \Qi{W} &\leq\1, \bigr. \\
                                   \bigl. \Tr \rho \Qi{W} &\geq 1-\varepsilon \bigr\}.
\end{split}
\end{equation}
This serves as a parent quantity for two other coherence functions.  Namely, we have
\begin{align}
C^\varepsilon_H(\rho)             &=\min\{ D^\varepsilon_H(\rho\|\delta)\,:\,\delta\in\mc{I} \}, \\
\widetilde{C}^\varepsilon_H(\rho) &=\min\{ D^\varepsilon_H(\rho\|\sigma)\,:\,\sigma=\Delta(\sigma),\ \tr\sigma=1\}.
\end{align}
The function $\widetilde{C}^\varepsilon_H(\rho)$ characterizes the one-shot distillable distillable coherence under MIO and DIO.

In the many-copy scenario, the \yunchao{regularized coherence distillation rate}
%asymptotic rate of coherence distillation
by operations $\mathcal{O}\in\{\text{MIO},\,\text{DIO},\,\text{IO},\,\text{SIO}\}$ is defined as
\yunchao{
\begin{equation}\label{asymptoticdistill}
  C_{d,\mathcal{O}}^{\infty}(\rho)=\liminf_{\varepsilon\to 0^+}\liminf_{n\to\infty}\frac{1}{n}C_{d,\mathcal{O}}^\varepsilon(\rho^{\otimes n}).
\end{equation}
}
For MIO, DIO, and IO, the distillable coherence of $\rho$ is given by the relative entropy of coherence, \yunchao{which we call the asymptotic coherence distillation rate~\cite{Winter16}, where
\begin{equation}\label{MIODIOIOdistill}
    \lim_{\varepsilon\to 0^+}\lim_{n\to\infty}\frac{1}{n}C_{d,MIO/DIO/IO}^\varepsilon\left(\rho^{\otimes n}\right)=C_r(\rho).
\end{equation}
Combining Eq.~\eqref{MIODIOdilution}, \eqref{MIODIOIOdistill}, we can conclude that coherence is strongly converse under MIO/DIO.}  In contrast, as we show below, certain coherent states are undistillable under SIO.  A summary of the distillation/formation rates for different operations is given in Table~\ref{Table:results}.

\section{One-shot distillation \protect\\ under MIO and DIO}
\label{sec:dios-mio}
In this section, we review the results of \cite{Regula-one-shot}
on MIO and DIO distillation.

For a quantum channel
$\Lambda:\mathcal{L}(A)\to\mathcal{L}(B)$,
consider its Choi operator
\begin{equation}
  \label{Choi}
  \Gamma = (\id\ox\Lambda)\Phi
         = \sum_{i,j}\ketbra{i}{j}_R\otimes\Lambda(\ketbra{i}{j}_A),
\end{equation}
where $R$ is isomorphic to $A$. Let $\Gamma_{ij}=\Lambda(\ketbra{i}{j}_A)$.
By Choi's theorem, $\Lambda$ is completely positive if and only if
$\Gamma\geq 0$. Also, $\Lambda$ is trace preserving if and only if
$\Tr \Gamma_{ij} = \delta_{ij}$.

Now consider $\Lambda\in \text{MIO}$, meaning $\Lambda(\delta)\in\mathcal{I}$
for all $\delta\in\mathcal{I}$, which is equivalent to $\Lambda(\proj{i})\in\mathcal{I}$
for all $i$. For the Choi matrix this says $\Gamma_{ii}\in\mathcal{I}$ for all $i$.
For $\Lambda\in \text{DIO}$, from $\Lambda(\Delta(\ketbra{i}{j}))=\Delta(\Lambda(\ketbra{i}{j}))$ for all $i$ and $j$, we get $\delta_{ij}\Gamma_{ij}=\Delta(\Gamma_{ij})$.

Since $\Lambda(\rho)=\sum_{i,j}\rho_{ij}\Gamma_{ij}$ (by linearity), we have now,
\begin{equation}
  \label{fidelity}
  F(\Lambda(\rho),\Psi_M)^2 = \Tr [\Lambda(\rho)\Psi_M]
                          = \sum_{i,j}\rho_{ij}\Tr \Gamma_{ij}\Psi_M.
\end{equation}
%{\color{red} It should be $F(\Lambda(\rho),\Psi_M)^2=\Tr [\Lambda(\rho)\Psi_M]...$}
The one-shot distillation under MIO/DIO can thus be expressed as
\begin{equation}\begin{split}
  \label{problem}
  \max M \text{ s.t. } \sum_{i,j} \rho_{ij }\Tr \Gamma_{ij}\Psi_M &\geq 1-\varepsilon, \\
                        \sum_{i,j}\ketbra{i}{j}\otimes\Gamma_{ij} &\geq 0,\\
                                                  \Tr \Gamma_{ij} &= \delta_{ij},\\
   &\!\!\!\!\!\!\!\!\!\!\!\!\!\!\!\!\!\!\!\!
    \begin{cases} \phantom{\delta_{ij}}\Gamma_{ii}\in\mathcal{I} & \text{[MIO]},\\
                  \delta_{ij}\Gamma_{ij}=\Delta(\Gamma_{ij})     & \text{[DIO]}.
   \end{cases}
\end{split}\end{equation}

Suppose that the $\Gamma_{ij}$ satisfy the above constraints.
Consider the twirling transformation
\begin{equation}
  \label{transform}
  \widetilde{\Gamma}_{ij} = \frac{1}{M!}\sum_{\pi \in S_M}U_\pi\Gamma_{ij}U_\pi^\dag,
\end{equation}
where $S_M$ denotes the permutation group on $M$ objects,
and $U_\pi$ is the unitary representation of $\pi$ as a permutation matrix.
It is easy to see that then the $\widetilde{\Gamma}_{ij}$ also satisfy the constraints,
and lead to the same objective function, due to the permutation
invariance of $\Psi_M$.
Therefore, without loss of generality, we can assume that each of the
$\Gamma_{ij}$ is permutation-invariant and has hence the form
\begin{equation}\label{gammaform}
  \Gamma_{ij}=\alpha_{ij}\Psi_M+\beta_{ij}\frac{\1-\Psi_M}{M-1}.
\end{equation}
Furthermore, $\Tr \Gamma_{ij} = \alpha_{ij}+\beta_{ij}$.
We can see that both MIO and DIO lead to the following constraints:
When $i\neq j$, we have $\beta_{ij}=-\alpha_{ij}$ and when $i=j$,
we have $\alpha_{ii}=\frac{1}{M}$, $\beta_{ii}=1-\frac{1}{M}$.
For the matrices $A=(\alpha_{ij})_{i,j}$ and $B=(\beta_{ij})_{i,j}$, this
translates into the simple relation $B=\1-A$.
Notice that
\begin{equation}
  \label{gammapositive}
  \Gamma = \sum_{i,j}\alpha_{ij}\ketbra{i}{j}\otimes\Psi_M
            + \sum_{i,j}\beta_{ij}\ketbra{i}{j}\otimes\frac{\1-\Psi_M}{M-1},
\end{equation}
and we must satisfy $\bra{\psi}\Gamma\ket{\psi}\geq 0$ for all $\ket{\psi}$.
Consider the special case $\ket{\psi}=\ket{u}\otimes\ket{v}$;
by choosing $\ket{v}=\ket{\Psi_M}$, we learn that $A\geq 0$.
By choosing $\ket{v}$ orthogonal with $\ket{\Psi_M}$, we see likewise
that $B\geq 0$. Conversely, $A,B\geq 0$ implies that $\Gamma\geq 0$.

In conclusion, one-shot distillation under MIO/DIO can be solved
by the following SDP:
$C_{d,MIO/DIO}^\varepsilon(\rho)=\log M_{\text{opt}}$, with
\begin{equation}\begin{split}
  \label{sdp}
  M_{\text{opt}} = \max M \text{ s.t. } \Tr \rho^T A &\geq 1-\varepsilon, \\
                                             0\leq A &\leq \1,\ A_{ii}=\frac{1}{M}\,\forall i.
\end{split}\end{equation}

To express this result in terms of a suitable relative-entropy distance,
we recall the \emph{hypothesis testing relative entropy} introduced above \cite{Wang12,Dupuis-entropies},
\begin{equation}\begin{split}
%  \label{dh}
  D_H^\varepsilon(\rho\|\sigma)
     = -\log \min \bigl\{\Tr \sigma W \,:\, 0 \leq W &\leq\1, \bigr. \\
                                   \bigl. \Tr \rho W &\geq 1-\varepsilon \bigr\},
\end{split}\end{equation}
and the corresponding coherence measure
\begin{equation}
  \label{dhcoherence}
  C_H^\varepsilon(\rho) = \min_{\delta\in\mathcal{I}} D_H^\varepsilon(\rho\|\delta).
\end{equation}

We show that the coherence measure $C_H^\varepsilon(\rho)$ can be computed by SDP.
Notice that
\begin{equation}\begin{split}
  \label{maxmin}
  C_H^\varepsilon(\rho) &= -\log\max_{\delta\in\mathcal{I}}
                             \min_{\substack{W:0\leq W\leq\1 \\ \Tr \rho W \geq 1-\varepsilon}}
                              \Tr \delta W \\
                        &= -\log\min_{\substack{W:0\leq W\leq\1 \\ \Tr \rho W \geq 1-\varepsilon}}
                             \max_{\delta\in\mathcal{I}}
                              \Tr \delta W \\
                        &= -\log\min_{\substack{W:0\leq W\leq\1 \\ \Tr \rho W \geq 1-\varepsilon}}
                             \max_i W_{ii},
\end{split}\end{equation}
where in the second line we have appealed to the minimax theorem \cite{Sion},
noting that the objective function, $\tr\delta W$, is linear in each
of the two arguments, and that the domains of optimization are closed convex sets.

\yunchao{The third line in Eq.~\eqref{maxmin}} can now be expressed as an SDP as follows:
\begin{equation}\begin{split}
  \label{dhsdp}
  \mu_{\text{opt}}=\max \mu \text{ s.t. } \Tr \rho W &\geq 1-\varepsilon, \\
                            0\leq W &\leq \1,\ W_{ii} \leq \frac{1}{\mu}\,\forall i,
\end{split}\end{equation}
\yunchao{where $C_H^\varepsilon(\rho)=\log\mu_{\text{opt}}$}.
Notice that in Eq. (\ref{sdp}) we have a transpose on $\rho$,
while in Eq. (\ref{dhsdp}) we don't. We can simply change $A$ into $A^T$
in the SDP (\ref{sdp}) without changing its value, so that the two SDP has similar
form, except that one has an equality sign where the other other
has a ``$\leq$''. In \cite{Regula-one-shot}, it was shown that
the r.h.s.~of Eq. (\ref{sdp}) can be expressed in terms of $D_H^\varepsilon$,
as well:
\yunchao{
\begin{equation}
\label{Eq:Hypothesis-DIO}
\begin{split}
  C_{d,MIO/DIO}^\varepsilon(\rho)=\max \log M \\
  \text{ s.t. } \Tr \rho^T A &\geq 1-\varepsilon, \\
                            0\leq A &\leq \1,\ A_{ii}=\frac{1}{M}\,\forall i \\
&\!\!\!\!\!\!\!\!\!\!\!\!\!\!\!\!\!\!\!\!\!\!\!\!\!\!\!\!\!\!\!\!\!\!\!
    =  \min D_H^\varepsilon(\rho\|\delta)\\&\!\!\!\!\!\!\!\!\!\!\!\!\!\!\!\!\!\!\!\!\!\!\!\!\!\!\!\!\!\!\! \text{ s.t. }\delta\text{ diagonal and }\tr\delta=1 \\
&\!\!\!\!\!\!\!\!\!\!\!\!\!\!\!\!\!\!\!\!\!\!\!\!\!\!\!\!\!\!\!\!\!\!\!
    =: \widetilde{C}_H^\varepsilon(\rho),
\end{split}
\end{equation}
} where the minimization is crucially over Hermitian, but not \Qi{necessarily positive semidefinite} matrices $\delta$.

We record these findings in the following theorem.

\begin{theorem}[Regula \emph{et al.}~\cite{Regula-one-shot}]
\label{thm:miodio}
For any state $\rho$, the one-shot MIO- and DIO-distillable coherence
is given by
\begin{equation}
  \label{boundmio}
  C_{d,\text{DIO}}^\varepsilon(\rho)
     =    C_{d,\text{MIO}}^\varepsilon(\rho)
     =    \widetilde{C}_H^\varepsilon(\rho)
     \leq C_H^\varepsilon(\rho).
\end{equation}
\end{theorem}

\section{One-shot distillation under IO}
\label{sec:des-pudels-kern:io}
Now we come to the main contribution of the present paper, the extension
of the one-shot distillation protocols to the class IO.

To start, recall the definition of the min-relative entropy and
its smoothed version,
\begin{equation}\label{Dmin}
  \begin{split}
    D_{\min}(\rho\|\sigma)     &= -\log F(\rho,\sigma)^2, \\
    D_{\min}^\varepsilon(\rho) &=  \max_{\rho'\in B^\varepsilon(\rho)} D_{\min}(\rho'\|\sigma).
  \end{split}
\end{equation}
%where the $\varepsilon$-ball $B^\varepsilon(\rho)$ is understood with respect to the purified distance $P(\rho',\rho) = \sqrt{1-F(\rho',\rho)^2}$.
The \emph{smooth min-relative entropy of coherence} is defined as
\begin{equation}
  \label{Cmin}
  C_{\min}^\varepsilon(\rho)
    = \max_{\rho'\in B^\varepsilon(\rho)} \min_{\delta\in\mathcal{I}} D_{\min}(\rho\|\delta).
\end{equation}
One might wonder why we do not interchange min and max here, but this
is the version of the quantity that appears naturally both in the achievability
bound we will derive on Subsection~\ref{subsec:achieve}, and in the
upper bound in Subsection~\ref{subsec:converse}.

\subsection{An achievable lower bound}
\label{subsec:achieve}
We will generalize the protocol in \cite[Thm.~6]{Winter16} to
obtain a min-relative-entropic lower bound on the distillable
coherence. In the process, the privacy amplification aspect will
become even more apparent.

\begin{theorem}
  \label{thm:achieve}
  For an arbitrary state $\rho$ and $0<\varepsilon<1$,
  \begin{equation}
    C^{\varepsilon}_{d,IO}(\rho) \geq C^{\frac{\varepsilon}{2}-\eta}_{\min}(\rho) - 2\log\frac{1}{\eta},
  \end{equation}
  for any $0< \eta < \frac{\varepsilon}{2}$.
\end{theorem}
\begin{proof}
\yunchao{
In order to accomplish our proof, we need to introduce the conditional min/max entropy and their smoothed versions. For a bipartite quantum state $\rho^{AB}$, the min-entropy of $A$ conditioned on $B$ is defined as
\begin{equation}
    H_{\min}(A|B)_{\rho^{AB}}:=-\min_{\sigma^B}D_{\max}(\rho^{AB}\|\1^A\otimes\sigma^B),
\end{equation}
and the max-entropy of $A$ conditioned on $B$ is defined as
\begin{equation}
    H_{\max}(A|B)_{\rho^{AB}}:=-H_{\min}(A|C)_{\rho^{AC}},
\end{equation}
where $\rho^{ABC}$ is a purification of $\rho^{AB}$ and $\rho^{AC}=\Tr_{B}\rho^{ABC}$. It is proven that the max-entropy has an alternative form~\cite{Konig2009operational}
\begin{equation}\label{maxentropyalternative}
    H_{\max}(A|B)_{\rho^{AB}}:=-\min_{\sigma^B}D_{\min}(\rho^{AB}\|\1^A\otimes\sigma^B).
\end{equation}
Furthermore, the smoothed conditional min- and max-entropy are defined by
\begin{equation}
\begin{aligned}
    H_{\min}^{\varepsilon}(A|B)_{\rho^{AB}}&:=\max_{\rho'\in B^\varepsilon(\rho^{AB})}H_{\min}^{\varepsilon}(A|B)_{\rho'},\\
    H_{\max}^{\varepsilon}(A|B)_{\rho^{AB}}&:=\min_{\rho'\in B^\varepsilon(\rho^{AB})}H_{\max}^{\varepsilon}(A|B)_{\rho'}.
\end{aligned}
\end{equation}
In the following proof, we denote the system of interest as $A$, i.e. $\rho^A := \rho$, and the incoherent basis as $\{\ket{x}\}$. Denote $\psi:=\ket{\psi}\bra{\psi}$.}
Choose a purification of $\rho^A$,
\[
  \ket{\psi}^{AE} = \sum_x \sqrt{p_x}\ket{x}^A\ket{\psi_x}^E,
\]
where $\Tr_E\psi^{AE}=\rho^A$, and introduce the
dephased cq-state
\[
  \omega^{AE} = (\Delta\ox\id)\psi = \sum_x p_x \proj{x}^A \ox \psi_x^E.
\]

According to \cite[Cor.~5.6.1]{Renner-PhD} (actually, a slight adaptation
to get rid of a factor of $2$), for every
$\log M \leq H_{\min}^{\frac{\varepsilon}{2}-\eta}(A|E)_\omega - 2\log\frac{1}{\eta}$
there exists a function $G:\mathcal{X} \rightarrow [M]=\{1,2,\ldots,M\}$
such that the state
\[
  \Omega^{KE} := (G\ox\id)\omega^{AE}
               = \sum_x p_x \proj{G(x)}^K \ox \psi_x^E
\]
satisfies
\[
  \frac12 \left\|\Omega^{KE} - \tau_K \ox \sigma^E\right\|_1 \leq \frac{\varepsilon}{2},
\]
for $\tau^K = \frac{1}{M}\1_K$ and a suitable state $\sigma^E$
(which may be equal to $\omega^E$, but it doesn't have to be).
%{\color{red} I think this can be achieved by arbitrary $\sigma^E$, but in the following we only need to let $\omega^E$ be $\sigma^E$}
Hence, because of the well-known relation
between trace distance and fidelity \cite{Fuchs-vandeGraaf},
\[
  1-F(\rho,\sigma) \leq \frac12\|\rho-\sigma\|_1 \leq \sqrt{1-F(\rho,\sigma)^2},
\]
we have $F(\Omega^{KE},\tau_K\ox\sigma^E) \geq 1-\frac{\varepsilon}{2}$.

To turn this property, which encapsulates the uniformity and independence
of the ``key'' $K=G(X)$ from $E$, into a coherence distillation protocol,
consider the following purification of $\Omega^{KE}$:
\[
  \ket{\Omega}^{KAE} := \sum_x \sqrt{p_x} \ket{G(x)}^K \ket{x}^A \ket{\psi_x}^E,
\]
which can be obtained at from $\ket{\psi}^{AE}$ by applying the isometry
\[
  U:\ket{x}^A \longmapsto \ket{G(x)}^K\ket{x}^A.
\]
Crucially, $U$ is incoherent (even SIO).
For $\tau_K\ox\sigma^E$, on the other hand, we choose a purification
$\ket{\Phi}^{KL} \ox \ket{\zeta}^{EF}$, with the standard maximally
entangled state $\ket{\Phi} = \frac{1}{\sqrt{M}}\sum_{i=0}^{M-1}\ket{i}\ket{i}$.
By Uhlmann's theorem \cite{uhlmann1976transition},
there exists an isometry $V:A \hookrightarrow LF$ such that
\begin{equation}
  \label{eq:big-fidel}
  |\bra{\Phi}\bra{\zeta}\cdot(\1_{KE}\ox V)\ket{\Omega}| \geq 1-\frac{\varepsilon}{2}.
\end{equation}
If we had $\ket{\Phi}^{KL}\ket{\zeta}^{EF}$, we could clearly create
a maximally coherent state $\Psi_M$ on system $K$ by tracing out $F$,
and destructively measuring $L$ in the Fourier conjugate
basis $\ket{\widehat{\alpha}} := Z^\alpha\ket{\Psi_M}$. \yunchao{Here $Z$ is the phase unitary defined in Eq.~\eqref{phaseunitary}.}

With this, the protocol is clear: starting from $\rho^A$, first
apply $U$; then apply $V$, followed by tracing out $F$,
measuring $L$ in the Fourier basis $\{\ket{\widehat{\alpha}}\}_{\alpha = 0,\ldots,M-1}$;
finally, apply $Z^\alpha$ on $K$.
The first and the third step clearly are incoherent (even SIO); the second
seems suspicious, and indeed $V$ may not be incoherent at all, but notice that
we follow it by a destructive measurement, and these are all IO.
We can write Kraus operators of the resulting map $\Lambda:A\rightarrow K$
as follows:
\[
  M_{\alpha\beta} := \bigl(Z^\alpha_K\ox \bra{\widetilde\alpha}^L\bra{\beta}^F V\bigr) U,
\]
where $\{\ket{\beta}\}$ is an arbitrary basis of $F$.

To analyze the fidelity of the protocol, we pass to the purifications
and look at eq.~(\ref{eq:big-fidel}); by the monotonicity of the
fidelity under the CPTP maps $\Lambda$ and $\Tr_E$, we obtain
$F\bigl( \Lambda(\rho),\Psi_M\bigr) \geq 1-\frac{\varepsilon}{2}$ and
hence $F\bigl( \Lambda(\rho),\Psi_M\bigr)^2 \geq 1-\varepsilon$.
%(thus maybe I should modified the result in this theorem to $ C^{2\varepsilon}_{d,IO}(\rho) \geq C^{\varepsilon-\eta}_{\min}(\rho) - 2\log\frac{1}{\eta}$)}.
%({\color{red}The smooth for IO may not be correct as we should have $F\bigl( \Lambda(\rho),\Psi_M\bigr)^2 \geq (1-\varepsilon)^2$})

This shows that $\log M \approx H_{\min}^{\frac{\varepsilon}{2}-\eta}(A|E)_\omega - 2\log\frac{1}{\eta}$
is achievable. Now, introducing the purification
$\ket{\omega}^{ABE} = \sum_x \sqrt{p_x}\ket{x}^A\ket{x}^B\ket{\psi_x}$, we have
\begin{equation}
\begin{split}\label{eq:hmincmin}
  H_{\min}^\alpha(A|E)_\omega &= - H_{\max}^\alpha(A|B)_\omega \\
        &=    -\log \min_{\omega'\in B^\alpha(\omega)} \max_\sigma F(\omega',\1\ox\sigma)^2 \\
        &\!\!\!\!\!\!\!\!\!\!
         \geq -\log \min_{\substack{\omega'\in B^\alpha(\omega) \\ \text{max.\,correlated}}}
                      \max_\sigma F(\omega',\1\ox\sigma)^2 \\
        &\!\!\!\!\!\!\!\!\!\!
         =    -\log \min_{\substack{\omega'\in B^\alpha(\omega) \\ \text{max.\,correlated}}}
                      \max_{\sigma\in\mathcal{I}} F(\omega',\1\ox\sigma)^2 \\
        &\!\!\!\!\!\!\!\!\!\!
         =    -\log \min_{\rho'\in B^\alpha(\rho)} \max_{\delta\in\mathcal{I}} F(\rho',\delta)^2 \\
        &\!\!\!\!\!\!\!\!\!\!
         = \max_{\rho'\in B^\alpha(\rho)} \min_{\delta\in\mathcal{I}} D_{\min}(\rho'\|\delta) \\
        &\!\!\!\!\!\!\!\!\!\!
         = C_{\min}^\alpha(\rho).
\end{split}
\end{equation}
\yunchao{Here, the second line follows from Eq.~\eqref{maxentropyalternative}}, and the third line is motivated by the observation that
$\omega^{AB} = \sum_{xy} \sqrt{p_x p_y}\braket{\psi_y}{\psi_x} \ketbra{xx}{yy}$
is a maximally correlated state, so it is natural to impose the same
structure on $\omega'$; the fourth line follows from the $Z\ox Z^\dagger$-invariance
of maximally correlated states, so by the concavity of the fidelity we
can impose w.l.o.g.~the same structure on $\1\ox\sigma$, meaning that
$\sigma$ is diagonal. The rest is straightforward algebra.
\end{proof}

\medskip
Note that the equality in Eq.~(\ref{eq:hmincmin}) can also be achieved by
directly using the results by Coles \cite{Coles2012}, or \cite{Liu18}:
\begin{equation}
  H_{\min}(X|E)_{\rho_{XE}} = \min_{\delta\in\mathcal{I}} D_{\min}(\rho_A\|\delta).
\end{equation}

\begin{remark}
\label{rem:DIIO}
The CPTP map $\Lambda(\cdot)=\sum_{\alpha,\beta}M_{\alpha\beta}(\cdot)M_{\alpha\beta}^\dagger$
we constructed in the proof is not only IO but also DIO. To see this, first expand
\begin{align}
\Lambda(\op{x}{y})=\sum_{\alpha,\beta}Z^\alpha\op{G(x)}{G(y)}Z^{\alpha\dagger}\bra{\widetilde{\alpha}\beta}V\op{x}{y}V^\dagger\ket{\widetilde{\alpha}\beta}\notag
\end{align}
for any incoherent basis states $\ket{x}$ and $\ket{y}$ of system $A$.  The key observation is that fully dephased states are \Qi{invariant} under conjugation by $Z^\alpha$, and moreover this conjugation commutes with $\Delta$, i.e. $\Delta[Z^\alpha(\cdot)Z^{\alpha\dagger}]=\Delta(\cdot)$.  Hence if we dephase system $K$ after applying the map $\Lambda$, we find
\begin{align}
\Delta[\Lambda(\op{x}{y})]&=\delta_{G(x)G(y)}\op{G(x)}{G(y)}\notag\\
&\quad\times\sum_{\alpha,\beta}\bra{\widetilde{\alpha}\beta}V\op{x}{y}V^\dagger\ket{\widetilde{\alpha}\beta}\notag\\
&=\delta_{xy}\op{G(x)}{G(y)}=\Lambda[\Delta(\op{x}{y})],
\end{align}
where the second equality follows from the fact that $\{\ket{\widetilde{\alpha}\beta}\}_{\alpha,\beta}$ forms a complete basis.
Thus, $\Lambda$ is a DIO map.
\end{remark}

\begin{corollary}
  \label{thm:diioachieve}
  For an arbitrary state $\rho$ and $0<\varepsilon<1$,
  \begin{equation}
    C^{\varepsilon}_{d,DIIO}(\rho) \geq C^{\frac{\varepsilon}{2}-\eta}_{\min}(\rho) - 2\log\frac{1}{\eta},
  \end{equation}
  for any $0< \eta < \frac{\varepsilon}{2}$,
  where DIIO refers to the intersection of IO and DIO.
  \hfill$\blacksquare$
\end{corollary}

\subsection{Upper bound and comparison with MIO distillation}
\label{subsec:converse}
We have a partial converse theorem to Theorem~\ref{thm:achieve} which
can bound $C^\varepsilon_{d,IO}(\rho)$ from both sides, as follows.

\begin{theorem}
  \label{thm:upper}
  For an arbitrary state $\rho$ and $0<\varepsilon<1$,
  \begin{equation}
    C^\varepsilon_{d,IO}(\rho) \le C_{\min}^{\sqrt{\varepsilon(2-\varepsilon)}}(\rho).
  \end{equation}
 \end{theorem}
\begin{proof}
Due to the inclusion of the classes of operations, and
Theorems~\ref{thm:miodio} and~\ref{thm:achieve}, we have the first
four of the following (in)equalities:
\begin{equation}\begin{split}
  \label{summary}
  C_{\min}^{\frac{\varepsilon}{2}-\eta}(\rho)-2\log\frac{1}{\eta}
     &\leq C_{d,\text{IO}}^{\varepsilon}(\rho)          \\
     &\leq C_{d,\text{M/D\,IO}}^\varepsilon(\rho)
           = \widetilde{C}_H^\varepsilon(\rho)          \\
     &\leq C_H^\varepsilon(\rho)                      \\
     &\leq C_{\min}^{\sqrt{\varepsilon(2-\varepsilon)}}(\rho).
\end{split}\end{equation}

The last one also follows essentially from known facts, namely \cite[Prop.~4.2]{Dupuis-entropies}.
Note only that our definition of $D_H^\varepsilon$ differs from
\cite{Dupuis-entropies} by $\varepsilon \leftrightarrow 1-\varepsilon$, and
an additional term of $\log(1-\varepsilon)$ added. We made this choice for
easier comparison with the results from \cite{Regula-one-shot}.
With this in mind, \cite[Eq.~(51)]{Dupuis-entropies} reads
\[
  D_H^\varepsilon(\rho\|\sigma) \leq D_{\min}^{\sqrt{2\varepsilon}}(\rho\|\sigma),
\]
and looking at the last step of the proof, one observes
that $\sqrt{2\varepsilon}$ can be improved to $\sqrt{\varepsilon(2-\varepsilon)}$.
We can adapt the proof in \cite{Dupuis-entropies} to include the
minimization over $\delta\in\mathcal{I}$, according to the following Lemma~\ref{lemma:mindh-vs-minmin}. By applying it with $\mathcal{S}=\mathcal{I}$, and
maximizing over the $\sqrt{\varepsilon(2-\varepsilon)}$-ball on the
right hand side, we precisely obtain the last inequality in
Eq. \ref{summary}.
\end{proof}

\begin{remark}
\yunchao{Combining Theorem~\ref{thm:achieve} and Theorem~\ref{thm:upper}, we conclude that $C_{d,IO}^{\varepsilon}(\rho)\approx C_{\min}^{\varepsilon'}(\rho)$, with $\varepsilon'\in[\frac{\varepsilon}{2},\sqrt{(2-\varepsilon)}]$.}
\end{remark}

\begin{lemma}
  \label{lemma:mindh-vs-minmin}
  Let $\mathcal{S}$ be a closed convex set of states on a Hilbert space $A$,
  and $\rho$ a state. Then, for every $0<\varepsilon<1$ there exists a
  subnormalized density matrix $\rho'$ with
  $P(\rho,\rho') \leq \sqrt{\varepsilon(2-\varepsilon)}$, such that
  \[
    \min_{\sigma\in\mathcal{S}} D_H^\varepsilon(\rho\|\sigma)
        \leq \min_{\sigma\in\mathcal{S}} D_{\min}^{\sqrt{\varepsilon(2-\varepsilon)}}(\rho'\|\sigma).
  \]
\end{lemma}
\begin{proof}
The crucial observation is that due to the convexity of the sets of
operators, $\mathcal{S}$ and $\{0\leq W\leq\1\,:\,\tr\rho W \geq 1-\varepsilon\}$,
we can invoke the minimax theorem~\cite{Sion}, to obtain
\[
  \max_{\sigma\in\mathcal{S}}
          \min_{\substack{0\leq W\leq\1 \\ \tr\rho W \geq 1-\varepsilon}} \tr\sigma W
  = \min_{\substack{0\leq W\leq\1 \\ \tr\rho W \geq 1-\varepsilon}}
          \max_{\sigma\in\mathcal{S}} \tr\sigma W.
\]
Thus, there exists an optimizer $W_0$ of the second expression,
$0\leq W_0\leq\1$, $\tr\rho W_0 \geq 1-\varepsilon$, with
\begin{equation}
  \label{eq:M0}
  \min_{\sigma\in\mathcal{S}} D_H^\varepsilon(\rho\|\sigma)
     = \min_{\sigma\in\mathcal{S}} -\log\tr\sigma W_0.
\end{equation}
Following the the example of \cite[Prop.~4.2]{Dupuis-entropies},
we define a subnormalized state $\rho' = \sqrt{W_0}\rho\sqrt{W_0}$,
which we claim to be the sought-after object.

To start with, from optimality of $W_0$, we have $\tr\rho' = \tr\rho W_0 = 1-\varepsilon$,
hence from \cite[Lemma~A.3]{Dupuis-entropies} (see also \cite[Lemma 7]{Berta:uncertainty}),
$P(\rho',\rho) \leq \sqrt{1-(\tr\rho W_0)^2} = \sqrt{\varepsilon(2-\varepsilon)}$.

At the same time, choosing a purification of $\rho^A = \tr_B \proj{\varphi}^{AB}$,
we get a purification of $\rho'$ by letting
$\ket{\varphi'} = (\sqrt{W_0}\ox\1)\ket{\varphi}$. Conjugating the
inequality $\varphi \leq \1$ by $\sqrt{W_0}\ox\1$ this results in
$\varphi' \leq W_0 \ox \1$. Now, just as in the proof of \cite[Prop.~4.2]{Dupuis-entropies},
we employ the dual variational characterization of the fidelity,
\[
  F(\rho',\sigma)^2 = \min\tr\sigma Z \text{ s.t. } \varphi'\leq Z\ox\1.
\]
This implies, that $\tr\sigma W_0 \geq F(\rho',\sigma)$ in \eqref{eq:M0},
for all $\sigma\in\mathcal{S}$, and so
$\min_{\sigma\in\mathcal{S}} D_H^\varepsilon(\rho\|\sigma)
     \leq \min_{\sigma\in\mathcal{S}} -\log F(\rho',\sigma)^2$, as claimed.
\end{proof}

%\bigskip
%\textcolor{red}{Restate the inequalities as theorem/corollary?}

\section{Distillation under SIO}
\label{sec:sio}

\subsection{Characterizing one-shot SIO distillation}

%A CPTP map $\Lambda:\mc{L}(A)\to\mc{L}(A)$ is a strictly incoherent operation (SIO) if it can be represented by Kraus operators $\{K_\alpha\}_\alpha$ in which
%\begin{equation}
%\label{Eq:SIO-Kraus}
%K_\alpha=\sum_{x}c_{\alpha,x}\op{f_\alpha(x)}{x},
%\end{equation}
%where $f_\alpha$ is a permutation on the set $\{1,\cdots,d_A\}$ and the $c_{\alpha,x}$ are complex numbers, possibly equaling zero.  The completion condition requires $\sum_\alpha|c_{\alpha,x}|^2=1$ for all $x$.  We recall the following entropic quantities.

We now turn to coherence distillation under strictly incoherent operations
(SIO). Ever since \cite{Winter16}, it has been an open question whether
coherence distillation such as the protocol in Sec. \ref{subsec:achieve},
or in \cite[Thm.~6]{Winter16}, really requires IO, or can be performed
within the smaller class of SIO (as all other protocols discussed in \cite{Winter16}
can).
The crucial object in this setting turns out to be the incoherent rank.
Recall that the incoherent rank of a positive operator $\Omega$ is defined by
\begin{equation}
  C_0(\Omega)=\min_{\{\lambda_j,\ket{\phi_j}\}}\max_j \log\text{rank}[\Delta(\phi_j)],
\end{equation}
where the minimization is taken over all positive rank-one decompositions of $\Omega$.
\begin{theorem}
\label{th:SIO}
For any state $\rho$, the one-shot distillable coherence under SIO is given by
\begin{align}
C_{d,SIO}^\varepsilon(\rho)=\max \log M \text{ s.t. } \Tr \rho A &\geq 1-\varepsilon, \notag \\
                            0\leq A &\leq \1,\ A_{ii}=\frac{1}{M}\,\forall i\notag\\
C_0(A)&\leq \log M .
\end{align}
\end{theorem}

\begin{proof}
Suppose that $\tr [\Lambda(\rho)\Psi_M] \geq 1-\varepsilon$ for some SIO map
$\Lambda$ and $\ket{\Psi_M}=\frac{1}{\sqrt{M}}\sum_{i=1}^M\ket{i}$.
Let $\Pi_{>M}=\frac{1}{M}\sum_{x=M+1}^{d_A}\op{x}{x}$.  Notice that
\begin{align}
 1-\varepsilon &\leq \tr [\Lambda(\rho)\Psi_M] \notag\\
            &\leq \tr [\Lambda(\rho)(\Psi_M+\Pi_{>M})] \notag\\
            &=    \tr [\rho A],
\end{align}
where $A:=\Lambda^*(\Psi_M)+\Lambda^*(\Pi_{>M})$ and \Qi{$\Lambda^*$ is the adjoint channel of $\Lambda$.}  Using the form of SIO Kraus operators, we have
\begin{align}
\label{Eq:dual-action}
\Lambda^*(\Psi_M)
 &=\frac{1}{M}\sum_\alpha \sum_{\substack{x\;s.t.\\f_\alpha(x)\in[M]}}\sum_{\substack{x'\;s.t.\\f_\alpha(x')\in[M]}}c_{\alpha,x}^*\op{x}{x'}c_{\alpha,x'}\notag\\
 &=\frac{1}{M}\sum_\alpha\op{\phi_\alpha}{\phi_\alpha},
\end{align}
and likewise
\begin{equation}
\Lambda^*(\Pi_{>M})=\frac{1}{M}\sum_\alpha\sum_{\substack{x\;s.t.\\f_\alpha(x)\not\in[M]}}|c_{\alpha,x}|^2\op{x}{x}.
\end{equation}
Thus $A$ has a decomposition into rank-one vectors each having an incoherent rank no greater than $M$.  Also, for any $y\in\{1,\cdots,d_A\}$, we see that
\begin{align}
\bra{y}A\ket{y}&=\frac{1}{M}\left(\sum_{\substack{\alpha\;s.t.\\f_\alpha(y)\in[M]}}+\sum_{\substack{\alpha\;s.t.\\f_\alpha(y)\not\in[M]}}\right)|c_{\alpha,y}|^2
 =\frac{1}{M},
\end{align}
i.e.~$A_{yy}=\frac{1}{M}$ for all $y$.

The converse involves essentially reversing these steps.  Suppose that
$\tr\rho A \geq 1-\varepsilon$ for some operator $0\leq A\leq \1$ with
$C_0(A)\leq \log M$ and $A_{ii}=\frac{1}{M}$.
Then there exists a decomposition
\begin{align}
 A&=\frac{1}{M}\sum_\alpha\op{\phi_\alpha}{\phi_\alpha}\notag\\
  &=\frac{1}{M}\sum_\alpha\sum_{x,x'=1}^{d_A}c_{\alpha,x}c^*_{\alpha,x'}\op{x}{x'},
\end{align}
where $(c_{\alpha,x})_x$ contains at most $M$ nonzero elements for each $\alpha$.  Hence we can define permutations $f_\alpha$ on the set $\{1,\cdots,|A|\}$ such that
$f_\alpha(x)\in[M]$ for every $x$ and $\alpha$ with $c_{\alpha,x}\not=0$.
The Kraus operators $K_\alpha=\sum_{x=1}^{d_A}c^*_{\alpha,x}\op{f_\alpha(x)}{x}$ satisfy
\begin{align}
\sum_\alpha K_\alpha^\dagger\Psi_M K_\alpha=\Omega.
\end{align}
Furthermore,
$\sum_\alpha K^\dagger_\alpha K_\alpha=\sum_{\alpha}\sum_{x=1}^{d_A}|c_{\alpha,x}|^2\op{x}{x}= \1$,
since by assumption
\[
  \frac{1}{M}=\bra{x}A\ket{x}=\frac{1}{M}\sum_\alpha|c_{\alpha,x}|^2.
\]
Therefore, the $\{K_\alpha\}_\alpha$ define a CPTP SIO  map
$\Lambda$ satisfying $\tr \Lambda(\rho)\Psi_M \geq 1-\varepsilon$.
\end{proof}

\begin{remark}
Comparing with Eq. (\ref{Eq:Hypothesis-DIO}), we see that the one-shot
distillable coherence under SIO takes the form of DIO distillation with
the added constraint of $C_0(A)\leq\log M$.
\end{remark}

\begin{remark}
An explicit calculation of the incoherent rank $C_0$ can be made through semi-definite programming techniques \cite{Ringbauer-2017a}.  However the number of computational constraints scales as $\binom{d}{M}$ for certifying whether a $d$-dimensional state $\rho$ has $C_0(\rho)\geq\log (M+1)$.
\end{remark}

\subsection{Bound coherence exists under SIO}
\label{subsec:bound-SIO}
The constraint on the incoherent rank of $A$ in Theorem \ref{th:SIO}
greatly diminishes the power of SIO to distill coherence.
Here we illustrate this effect by a dramatic example. Consider the state
\begin{equation}\label{boundstate}
  \rho = \frac12 \bigl( \proj{+}\otimes\proj{+} + \proj{-}\otimes\proj{\widetilde{+}} \bigr),
\end{equation}
where
\begin{align}
\ket{+}\ket{+}&=\frac{1}{2}(\ket{00}+\ket{01}+\ket{10}+\ket{11})\notag\\
\ket{-}\ket{\widetilde{+}}&=\frac{1}{2}(\ket{00}+i\ket{01}-\ket{10}-i\ket{11}).\notag
\end{align}
The $n$-copy state $\rho^{\otimes n}$ is then an equal mixture of states belonging
to the ensemble
\[
  \mf{E}_n:=\{\ket{+}\ket{+},\ket{-}\ket{\widetilde{+}}\}^{\otimes n}.
\]

\Qi{We will show that not a single cosbit of coherence can be distilled
from $\rho^{\otimes n}$ by SIO with an error smaller than the minimal
one-copy error.  In comparison, $n$ bits of coherence can be distilled by IO error-free: the first system in each copy of $\rho$ is simply measured
with the IO Kraus operations $\{\op{0}{+},\op{1}{-}\}$ followed by a suitable
controlled phase on the second qubit.
Such a measurement is not possible by SIO.}

\begin{theorem}
\label{th:boundSIO}
For the state $\rho$ \Qi{defined in Eq.~\eqref{boundstate}}, $C_{d,SIO}^{\infty}(\rho) = 0$.
\end{theorem}

The proof of this will follow by studying the structure of $\rho^{\otimes n}$ and showing that for a fixed value of $\epsilon$ (independent of $n$), $\tr\rho^{\otimes n} A<1-\epsilon$ for any operator $A$ having an incoherent rank of two and satisfying the conditions of Theorem \ref{th:SIO}.  The key property we use is that the eigenvectors of $\rho^{\otimes n}$ will always be maximally coherent states with complex phases belonging to $\{0,\frac{\pi}{2},\pi,\frac{3\pi}{2}\}$.   For the $n$-copy analysis to be tractable, we need to introduce some new notation.  Let $\mbf{b}_j=(b_{j,0},b_{j,1},\cdots,b_{j,n-1})\in\{0,1\}^{n}$ denote the $j^{\text{th}}$ binary sequence of length $n$.  We then define an ensemble of $2^n$ equiprobable states $\{\ket{\mbf{b}_j}\}_{j=1}^{2^n}$, where
\begin{align}
&\ket{\mbf{b}_{j}}:=\notag\\
&\frac{1}{\sqrt{4^n}}\sum_{m_0,\cdots,m_{n-1}=0}^3\exp\left(i\frac{\pi}{2}\sum_{k=0}^{n-1}b_{j,k}m_k\right)\bigg|\sum_{k=0}^{n-1}4^km_k\bigg\rangle.
\end{align}
We claim that, up to relabelling, this ensemble is precisely $\mf{E}_n$.  For example, when $n=1$ we have
\begin{align}
\ket{\mbf{b}_1}&=\frac{1}{2}\left(\ket{0}+\ket{1}+\ket{2}+\ket{3}\right)\notag\\
\ket{\mbf{b}_2}&=\frac{1}{2}\left(\ket{0}+i\ket{1}-\ket{2}-i\ket{3}\right),\notag
\end{align}
and for $n=2$ we have

%\begin{align}
%\ket{\mbf{b}_1}&=\frac{1}{4}\left(\ket{0}+\ket{1}+\ket{2}+\ket{3}+\ket{4}\notag\\
%&+\ket{5}+\ket{6}+\cdots+\ket{14}+\ket{15}\right)\notag
%\end{align}

%\begin{widetext}
\begin{equation*}
\begin{aligned}
\ket{\mbf{b}_1}&=\frac{1}{4}(\ket{0}+\ket{1}+\ket{2}+\ket{3}+\ket{4}\\
&+\ket{5}+\ket{6}+\cdots+\ket{14}+\ket{15})\\
\ket{\mbf{b}_2}&=\frac{1}{4}(\ket{0}+i\ket{1}-\ket{2}-i\ket{3}+\ket{4}\\
&+i\ket{5}-\ket{6}-\cdots-\ket{14}-i\ket{15})\\
\ket{\mbf{b}_3}&=\frac{1}{4}(\ket{0}+\ket{1}+\ket{2}+\ket{3}+i\ket{4}\\
&+i\ket{5}+i\ket{6}+\cdots-i\ket{14}-i\ket{15})\\
\ket{\mbf{b}_4}&=\frac{1}{4}(\ket{0}+i\ket{1}-\ket{2}-i\ket{3}+i\ket{4}\\
&-\ket{5}-i\ket{6}+\cdots+i\ket{14}-\ket{15}).
\end{aligned}
\end{equation*}
%\end{widetext}
The case of general $n$ can be checked by induction.

Now for any two distinct vectors $\vec{m}=(m_0,\cdots,m_{n-1})$ and $\vec{m}'=(m_0',\cdots,m_{n-1}')$ belonging to $\{0,1,2,3\}^{n}$, let us denote the kets
\begin{align}
\ket{\vec{m}}&:=\bigg|\sum_{k=0}^{n-1}4^km_k\bigg\rangle,&\ket{\vec{m}'}&:=\bigg|\sum_{k=0}^{n-1}4^km'_k\bigg\rangle.
\end{align}
The relative phase between $\ket{\vec{m}}$ and $\ket{\vec{m}'}$ in any $\ket{\mbf{b}_j}$ is defined to be
\begin{equation}
\frac{\pi}{2}\sum_{k=0}^{n-1}b_{j,k}(m_k'-m_k)\in\left\{0,\frac{\pi}{2},\pi,\frac{3\pi}{2}\right\}.
\end{equation}
We now make a crucial observation about the distribution of relative phases among the $\ket{\mbf{b}_j}$ in $\mf{E}_n$.
\begin{proposition}
For any fixed pair of distinct vectors $\vec{m}$ and $\vec{m}'$, at most half of the $\ket{\mbf{b}_i}$ in $\mf{E}_n$ have the same relative phase between  $\ket{\vec{m}}$ and $\ket{\vec{m}'}$.
\end{proposition}
\begin{proof}
Let $k\in\{0,\cdots,n-1\}$ be chosen such that $m_{k}'-m_{k}\not=0$.  Let $\Delta m=m_k'-m_k$.  Consider all the states $\ket{\mbf{b}_j}\in\mf{E}_n$ with binary sequence $\mbf{b}_j$ such that $b_{j,k}=0$; this represents exactly half of all states in the ensemble $\mf{E}_n$.  We partition these states into four groups, $A_0$, $A_{\pi/2}$, $A_{\pi}$, and $A_{3\pi/2}$, according to their respective relative phases between $\ket{\vec{m}}$ and $\ket{\vec{m}'}$.  Now we consider the other half of the states in $\mf{E}_n$, those having $b_{j,k}=1$.  We likewise partition these states into sets $B_0$, $B_{\pi/2}$, $B_{\pi}$, and $B_{3\pi/2}$ of common relative phase between $\ket{\vec{m}}$ and $\ket{\vec{m}'}$.  Notice that for any $\ket{\mbf{b}_j}$ in, say, $A_{\pi/2}$, there will be a corresponding $\ket{\mbf{b}_{j'}}$ in $B_{\pi/2(1+\Delta m})$ and vice versa, the only difference between $\mbf{b}_j$ and $\mbf{b}_{j'}$ being their $k^{th}$ component.  Hence $|A_0|=|B_{\pi/2\Delta m}|$, $|A_{\pi/2}|=|B_{\pi/2(1+\Delta m)}|$, $|A_\pi|=|B_{\pi/2(2+\Delta m)}|$ and $|A_{3\pi/2}|=|B_{\pi/2(3+\Delta m)}|$, where all arithmetic is done modular $4$.  Therefore, the total number of states  in the ensemble having a relative phase of, say, $\pi/2$ is
\begin{align}
|A_{\pi/2}|+|B_{\pi/2}|&=|A_{\pi/2}|+|A_{\pi/2(1-\Delta m)}|\leq 2^{n-1}.\notag
\end{align}
The same bound likewise holds for the other three relative phases.
\end{proof}

\begin{proof-of}{Theorem~\ref{th:boundSIO}}
Consider an arbitrary \Qi{vector in the complex linear span of $\ket{\vec{m}}$, $\ket{\vec{m}'}$},
\begin{equation}
\ket{\psi}=\cos\theta\ket{\vec{m}}+\sin\theta e^{i\phi}\ket{\vec{m}'}.
\end{equation}
Let $N_0$, $N_{\pi/2}$, $N_\pi$, $N_{3\pi/2}$ denote the number of states in $\mf{E}_n$ having a relative phase between $\ket{\vec{m}}$ and $\ket{\vec{m}'}$ of $0$, $\pi/2$, $\pi$, and $3\pi/2$, respectively.  We can then explicitly compute
\begin{align}
&\bra{\psi}\rho^{\otimes n}\ket{\psi}
 =\frac{1}{2^n4^n}\sum_{i=1}^{2^n}|\ip{\psi}{\mbf{b}_i}|^2\notag\\
&=\frac{N_0}{2^n4^n}|\cos\theta\!+\!e^{-i\phi}\sin\theta|^2
  +\frac{N_{\pi/2}}{2^n4^n}|\cos\theta\!+\!ie^{-i\phi}\sin\theta|^2\notag\\
&\phantom{=}
  +\frac{N_\pi}{2^n4^n}|\cos\theta\!-\!e^{-i\phi}\sin\theta|^2
  +\frac{N_{3\pi/2}}{2^n4^n}|\cos\theta\!-\!ie^{-i\phi}\sin\theta|^2\notag\\
&=\frac{1}{4^n}
  +\frac{\sin2\theta}{2^n4^n}\left[(N_0\!-\!N_\pi)\cos\phi+(N_{3\pi/2}\!-\!N_{\pi/2})\sin\phi\right],\label{Eq:N-copy-overlap}
\end{align}
where the last line follows by expanding out the squared amplitudes, the identity
$2\cos\theta\sin\theta=\sin2\theta$, and using the
fact that $N_0+N_{\pi/2}+N_\pi+N_{3\pi/2}=2^n$.  Our goal is to
maximize \eqref{Eq:N-copy-overlap} under the constraint that
$\max\{N_0,N_{\pi/2},N_\pi,N_{3\pi/2}\}\leq 2^{n-1}$.
This constraint implies that $|N_0-N_\pi|\leq 2^{n-1}$ and $|N_{3\pi/2}-N_{\pi/2}|\leq 2^{n-1}$,
and so
\begin{align}
\bra{\psi}\rho^{\otimes n}\ket{\psi}&\leq \frac{1}{4^n}\left(1+\cos\theta\sin\theta(|\cos\phi|+|\sin\phi|)\right)\notag\\
&\leq \frac{1}{4^n}\left(1+\frac{\sqrt{2}}{2}\right).
\end{align}
Suppose now that $A$ has an incoherent rank of two and satisfies $\tr[A]=\frac{4^n}{2}$.  Then by the previous calculation we have the fidelity bound
\begin{align}
\tr \rho^{\otimes n}A \leq (\tr A) \frac{1}{4^n}\left(1+\frac{\sqrt{2}}{2}\right)
                    &=   \frac{1}{2}\left(1+\frac{\sqrt{2}}{2}\right)\notag\\
                    &=   1-\varepsilon,
\end{align}
where $\varepsilon=\frac{1}{2}-\frac{\sqrt{2}}{4}$ is independent of $n$.
This is precisely the single-copy error bound.
As a consequence, it follows that $C_{d,SIO}^\infty(\rho)=0$,
proving the theorem.
\end{proof-of}

\begin{remark}
This result should be compared with the recent proof,
by Marvian \cite{Marvian:no-go}, that coherence distillation is
generally impossible in the resource theory of energy conservation,
which is characterized by the class of so-called time-translation-covariant
operations (TIO). That class is difficult to compare with DIO, as
at the single-system level, TIO is contained in DIO, but since the
composition of systems works differently, it may result in TIO operations
outside DIO on the multi-system level.

The result of \cite{Marvian:no-go} shows that for generic mixed states,
the rate of distilling cosbit states $\ket{\Psi_2}$ is zero, but
that at the same time it is possible to obtain a single cosbit (or a sublinear
number) with fidelity going to $1$ as asymptotically many copies of the
mixed resource become available.
In contrast, here we showed that under DIO the fidelity remains bounded
away from $1$, irrespective of the number of resource states.
\end{remark}

\section{Recovering the information theoretic limit}
\label{sec:AEP}
In the asymptotic limit, the coherence distillation rate under operation class
$\mathcal{O}$ is defined as
\yunchao{
\begin{equation}
%\label{asymptoticdistill}
  C_{d,\mathcal{O}}^{\infty}(\rho) = \liminf_{\varepsilon \to 0^+}\liminf_{n\to\infty}\frac{1}{n}C_{d,\mathcal{O}}^\varepsilon(\rho^{\otimes n}).
\end{equation}}
From \cite{Winter16} and \cite{Regula-one-shot} (see also \cite{Eric:MIO-CoF}) we know that
$C_{d,IO}^{\infty}(\rho) = C_{d,DIO}^{\infty}(\rho) = C_{d,MIO}^{\infty}(\rho) = C_r(\rho)$.
Below we show that our results on one-shot IO distillation can be used
to recover the asymptotic limit, at the same time improving the result by showing
that the limit exists and equals $C_r(\rho)$ for any fixed $0<\varepsilon<1$;
such a statement is known as a \emph{strong converse} in information theory.

\begin{theorem}
\label{thm:asymptotic-IO}
For any state $\rho$ and any $0<\varepsilon<1$,
\begin{equation*}\begin{split}
  \lim_{n\to\infty}\frac{1}{n}C_{d,IO}^\varepsilon(\rho^{\otimes n})
     &= \lim_{n\to\infty}\frac{1}{n}C_{d,DIO}^\varepsilon(\rho^{\otimes n}) \\
     &= \lim_{n\to\infty}\frac{1}{n}C_{d,MIO}^\varepsilon(\rho^{\otimes n})
      = C_r(\rho).
\end{split}\end{equation*}
\end{theorem}

\begin{proof}
Recall the results in Theorems \ref{thm:achieve} and
\ref{thm:upper}, which state that for any $\eta < \frac12 \varepsilon$,
\begin{equation}\begin{split}
  \label{one-shot-IO}
  C_{\min}^{\frac{\varepsilon}{2}-\eta}(\rho)-2\log\frac{1}{\eta}
     &\leq C_{d,IO}^{\varepsilon}(\rho) \\
     &\leq C_{d,MIO}^{\varepsilon}(\rho) \\
     &=    C_{d,DIO}^{\varepsilon}(\rho)
      \leq C_{\min}^{\sqrt{\varepsilon(2-\varepsilon)}}(\rho).
\end{split}\end{equation}
Hence, to show the theorem, we only need to prove that for all $0<\delta<1$,
\begin{equation}
  \lim_{n\to\infty}\frac{1}{n}C_{\min}^{\delta}(\rho^{\otimes n})=C_r(\rho),
\end{equation}
which is equivalent to
\begin{equation}\label{asymptotic-IO-Hmin}
  \lim_{n\to\infty}\frac{1}{n}\max_{\rho'\in B_\delta({\rho^A}^{\otimes n})} H_{\min}(X^n|E^n)_{\omega'}
     =C_r(\rho).
\end{equation}
Here, $\omega_{X^nE^n}'=(\Delta_{A^n}\ox\id_{E^n})\proj{\psi'}$
and $\ket{\psi'}$ is an arbitrary purification of $\rho_A'$.
Recall the quantum asymptotic equipartition theorem \cite{tomamichel2009fully},
which states that for any $0<\eta<1$,
\begin{equation}
  \lim_{n\to\infty}\frac{1}{n}\max_{\rho_{AB}'\in B_\eta(\rho_{AB}^{\otimes n})} H_{\min}(A|B)_{\rho_{AB}'}
       =H(A|B)_{\rho_{AB}}.
\end{equation}

Now, in one direction, if we have a state $\rho'\in B_\delta(\rho^{\otimes n})$,
then by Uhlmann's characterization of the fidelity
there exists a purification $\psi'\in B_\delta(\psi^{\otimes n})$,
hence $\omega' = (\Delta_{A^n}\ox\id_{E^n})\psi' \in B_\delta(\omega^{\otimes n})$,
with $\omega^{XE} = (\Delta_A\ox\id_E)\psi$.
In the other direction, for $\omega''\in B_\delta(\omega^{\otimes n})$,
it is known that since $\omega$ is a cq-state, an optimal
$\omega''$ for $H_{\min}(X^n|E^n)$
may be assumed to be a cq-state as well \cite{Renner-PhD},
hence we can find a $\psi'\in B_\delta(\psi^{\otimes n})$
such that $\omega'' = (\Delta_{A^n}\ox\id_{E^n})\psi'$.
Thus, we can conclude that
\[
  \max_{\rho_A'\in B_\delta({\rho^A}^{\otimes n})} H_{\min}(X^n|E^n)_{\omega'}
    = \max_{\omega''\in B_\delta(\omega^{\otimes n})} H_{\min}(X^n|E^n)_{\omega''},
\]
where the left hand side corresponds to Eq. (\ref{asymptotic-IO-Hmin}).

But this means that we can apply the quantum AEP directly, and get
\[\begin{split}
  \lim_{n\to\infty}\frac{1}{n}C_{\min}^{\delta}(\rho^{\otimes n})
    &= \lim_{n\to\infty}\frac{1}{n}\max_{\omega'\in B_\delta({\omega}^{\otimes n})} H_{\min}(X^n|E^n)_{\omega'} \\
    &= H(X|E)_{\omega} = C_r(\rho),
\end{split}\]
as claimed.
\end{proof}

\section{Coherence distillation \protect\\and randomness generation}
\label{sec:randomness}
Suppose a purification of $\rho^A$ is written as
\[
  \ket{\psi}^{AE} = \sum_x \sqrt{p_x}\ket{x}^A\ket{\psi_x}^E,
\]
we use this state to generate randomness by first performing a
computational basis measurement.
The dephased cq-state after measuring $A$ in the computational basis is
\[
  \omega^{AE} = (\Delta\ox\id)\psi = \sum_x p_x \proj{x}^A \ox \psi_x^E.
\]
Considering the measurement as a raw randomness generation process,
a subsequent randomness extraction (via a deterministic function $G$)
can further extract a random string that is almost uniform and independent of $E$.
We think of the function as an incoherent operation,
by letting $G(\ketbra{x}{y}) = \delta_{xy}\proj{G(x)}$. This identification
is natural as every incoherent (MIO) operation $\Lambda$ defines an
associated classical channel via $\Lambda(\proj{x}) = \sum_y \Lambda(y|x)\proj{y}$.

Denote $\ell^\varepsilon_{\textrm{ext}}$ to be the maximum length of the extractable randomness
that is $\varepsilon$-close to a string that it is perfectly uniform and independent of $E$, i.e.
\begin{equation}
  \ell^\varepsilon_{\textrm{ext}}(\rho_A)
   = \max_{G}\left\{\log M: \frac{1}{2}\left\|(G\!\ox\!\id)\omega^{AE}\!-\!\tau^K\!\otimes\!\omega^E\right\|_1\le \varepsilon\right\}\!,
\end{equation}
where we recall the notation $\tau^K = \frac{1}{M}\1_K$ for the maximally
mixed state of the $M$-dimensional key system.
%\textcolor{red}{We should discuss how much this model resembles the
%one discussed by Hayashi and Zhu, Phys. Rev. A 97:012302!}
%{\color{red}I am not sure whether $\sigma^E$ is an arbitrary state or $\Omega^E$.}
%and according to Theorem \ref{thm:achieve}
%and Ref.~\cite{Renner-PhD},
%\begin{equation}
%	\ell^\varepsilon_{\textrm{ext}} = C^{\varepsilon-\eta}_{\min}(\rho) - 2\log\frac{1}{\eta}.
%\end{equation}

%Note that the definition of extractable randomness here is different with that in \cite{PhysRevA.97.012302}. \cite{PhysRevA.97.012302} focus on the extractable randomness with incoherent strategies, as shown in Fig.~\ref{fig:extract}(a).
Note that our definition of extractable randomness
differs somewhat from the one in \cite{PhysRevA.97.012302}; in that work, a
model based on incoherent operations was proposed, which
is shown in Fig.~\ref{fig:extract}(a).
The main process consists of three parts, incoherent operations $\Lambda$, dephasing operation $\Delta$, a random hashing function as the extractor $\mathrm{Ext}$. Here our definition is more straightforward and we do not need to perform the real incoherent operations. As shown in Fig.~\ref{fig:extract}(b), after the dephasing operation $\Delta$, we use a function $G$ as an extractor to extract the secure randomness. This function has to be deterministic, as opposed to a noisy channels, since otherwise infinite randomness can be generated independent of $E$.

Moreover, in order to obtain the optimal $G$ in our definition, we first consider the randomness extraction process via DIO which is shown in Fig.~\ref{fig:extract}(c), where we apply the DIO distillation followed by dephasing operation $\Delta$, a classical extractor (which may not be needed). Benefiting from the property of DIO, we can change the order of DIO and $\Delta$, which implies that the DIO distillation may act as a good extractor (the blue part in Fig.~\ref{fig:extract}(d)). The only remaining problem here is that we would have to show that this DIO operation gives rise to a deterministic classical channel, which is in general not true.
For instance, the optimal coherence distillation process under DIO derived in Theorem~\ref{thm:miodio}
has the property that $\Lambda(\proj{x}) = \tau_K$ for all $x$. Via the permutation twirling Eq.~(\ref{transform}), this can be imposed equally on any optimal IO distillation process, and hence
even on DIIO = IO$\cap$DIO.

Instead, inspired by Remark \ref{rem:DIIO}, we know that a suboptimal but achievable IO distillation operation is also a DIO operation and after a modification (another dephasing channel) we can construct a valid extractor $G$ from it, which is shown in the proof of Theorem~\ref{th:extractIO}.

\begin{figure}[bht]
\centering \resizebox{9cm}{!}{\includegraphics[scale=1]{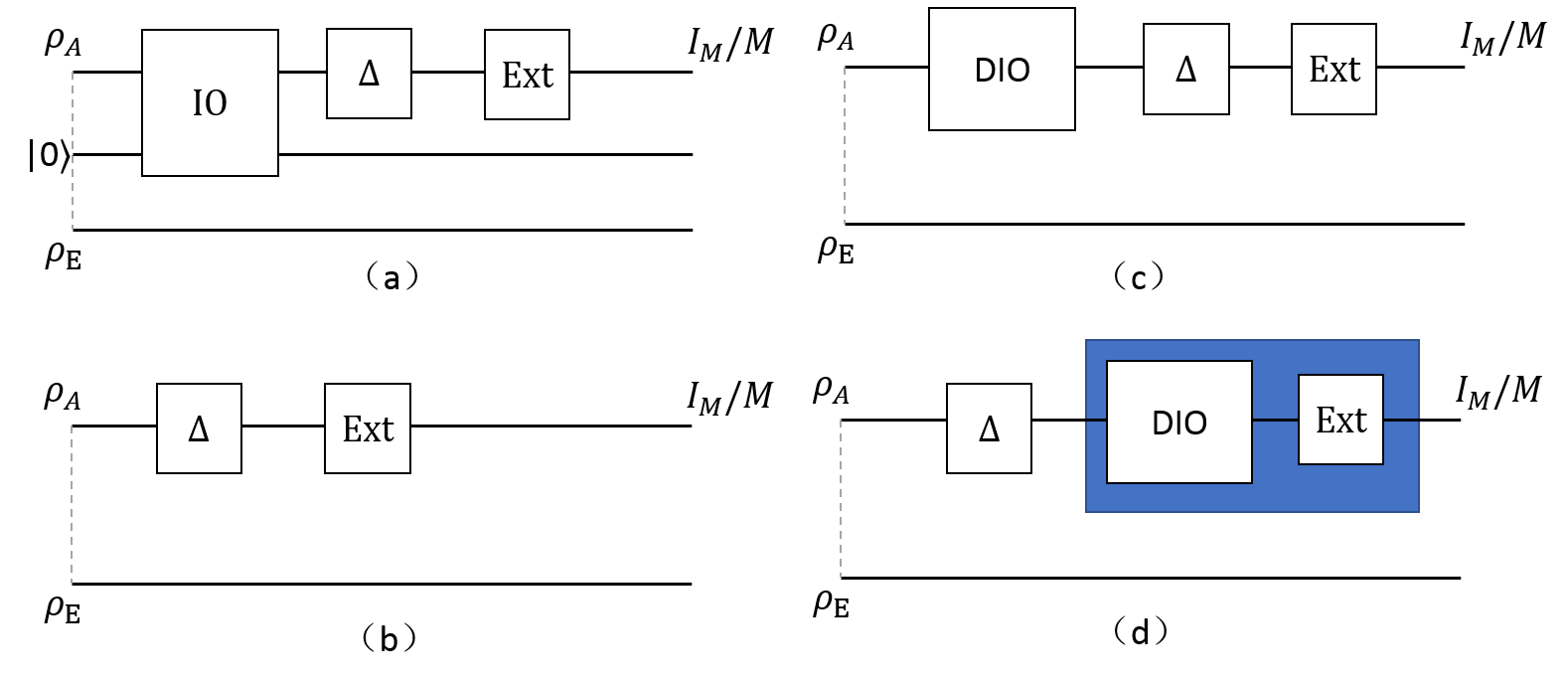}}
\caption{The different schemes for extracting randomness. (a) Extracting randomness using incoherent strategy in \cite{PhysRevA.97.012302}. Alice is allowed to perform the unitary IO operation on her system together with another ancilla system prepared in $\ket{0}\bra{0}$.
 (b) The randomness extraction process defined in our protocol. The extractor is implemented by a deterministic function $G$ (c) and (d) Applying DIO distillation for extracting secure randomness. Due to the property of DIO, we can achieve a new operation (blue part) combining the DIO distillation operation and original extractor. } \label{fig:extract}
\end{figure}

We first recall that every valid $G$ can be used for IO distillation.
Here, we consider $G$ as a deterministic function, and
$\omega^E=\psi^E$ is the reduced state of $\omega^{AE}$ on system E.
For every function $G$ satisfying
$\frac{1}{2}\left\|(G\ox\id)\omega^{AE}-\tau^K\otimes \omega^E\right\|_1\le \varepsilon$,
we can substitute it in the proof of
Theorem \ref{thm:achieve} and obtain an incoherent distillation channel.
Thus as the maximal distillable rate $C^\varepsilon_{d,IO}(\rho)$, it satisfies
\begin{equation}\label{extrac-Io}
  C^{2\varepsilon}_{d,IO}(\rho)\geq \ell^\varepsilon_{\textrm{ext}}(\rho).
\end{equation}

The IO map achieving $C^{2\varepsilon}_{d,IO}(\rho)$ can also be applied to
extract randomness by the following theorem.
%$\ell^\varepsilon_{\textrm{ext}}$ is also related to the distillable coherence
%under DIO,

\begin{theorem}\label{th:extractIO}
For an arbitrary state $\rho$ and $0<\varepsilon<1$,
\begin{equation}
  \ell^{\sqrt{2\varepsilon}}_{\textrm{ext}}(\rho)
      \geq C^{(\frac{\varepsilon}{2}-\eta)^2/2}_{d,IO}(\rho) - 2\log\frac{1}{\eta},
\end{equation}
for any $0< \eta < \frac{\varepsilon}{2}$.
\end{theorem}
\begin{proof}
Recall that the distillable coherence under IO is given by
\[
C_{d,IO}^\varepsilon(\rho)
        = \max_{\Lambda\in \mathcal{IO}}
          \bigl\{\log M \,:\, F(\Lambda(\rho),\Psi_M)^2 \geq 1-\varepsilon \bigr\}.
\]
Suppose $\Lambda$ is the IO that achieves the right hand side of Theorem \ref{thm:achieve},
that is $C^{\frac{\varepsilon}{2}-\eta}_{\min}(\rho) - 2\log\frac{1}{\eta} = \log M$ where $0< \eta < \frac{\varepsilon}{2}$ and
$F(\Lambda(\rho),\Psi_M)^2 \geq 1-\varepsilon$.
Note that $\Lambda$ is not necessarily an optimal IO coherence distillation operation.
%Note that $\Lambda$ is slightly different with the optimal IO distillation operation.
For the purification state $\ket{\psi}^{AE}$, the resulting state
by applying $\Lambda$ on system $A$ is given by
\[
 \Omega^{AE}=(\Lambda\otimes \id)(\ket{\psi}\bra{\psi}^{AE}),
\]
and it follows that
\[
 F(\Omega^{AE}, \Psi_M\ox \Omega^E)^2\ge (1-\varepsilon)^2.
\]
To prove that, suppose a purification of $\Omega^{AE}$ is $\ket{\Omega}^{AA'E}$.
Then, considering an orthogonal basis $\{\ket{\phi_x}^A\}$ of system $A$
such that $\ket{\phi_0}^A=\ket{\Psi_{M}}$, we can write
\[
  \ket{\Omega}^{AA'E} = \sum_{x}\alpha_x\ket{\phi_x}^A\ket{\psi_x}^{A'E}.
\]
As $F(\Lambda(\rho),\Psi_M)^2\geq 1-\varepsilon$ and
\[
  \Lambda(\rho) = \sum_{xx'}\alpha_x\alpha_x'\bra{\psi_{x'}}\psi_x\rangle^{A'E}\ket{\phi_x}\bra{\phi_{x'}}^A,
\]
we have
$F(\Lambda(\rho),\Psi_M)^2 = |\alpha_0|^2\ge 1-\varepsilon$.
The fidelity between $\Omega^{AA'E}$ and $\Psi_M^A\otimes\psi_0^{A'E}$ is
\[
  F({\Omega}^{AA'E},\Psi_M^A\otimes{\psi_0}^{A'E})^2 = |\alpha_0|^2\ge 1-\varepsilon.
\]
Denoting $\Omega^{A'E} = \tr_A {\Omega}^{AA'E}$, then the
fidelity between ${\Omega}^{AA'E}$ and $\Psi_M^A\otimes\Omega^{A'E}$ is
\[
  F({\Omega}^{AA'E},\Psi_M^A\otimes\Omega^{A'E})^2 = |\alpha_0|^4\ge (1-\varepsilon)^2.
\]
Then $ F(\Omega^{AE}, \Psi_M\ox \Omega^E)^2\ge (1-\varepsilon)^2 $ can be obtained by tracing out system $A'$.
 	
Applying the dephasing operation $\Delta$ on system $A$,
\[
 F\left((\Delta\ox\id)\Omega^{AE}, \tau^K\ox \Omega^E\right)^2\ge F(\Omega^{AE}, \Psi_M\ox \Omega^E)^2.
\]
From the Remark \ref{rem:DIIO}, we know that $\Lambda$ is also a DIO which
commutes with $\Delta$, so we have equivalently
\begin{equation}
\begin{aligned}
 F\left((\Lambda\ox\id)\omega^{AE}, \tau^K\ox \Omega^E\right)^2&\ge F(\Omega^{AE}, \Psi_M\ox \Omega^E)^2 \\
 &\ge (1-\varepsilon)^2.
\end{aligned}
\end{equation}
In order to construct a deterministic $G$, we apply another dephasing operation
after the incoherent channel $\Lambda$,
\begin{equation}
\begin{aligned}
 F\left((\Delta\circ\Lambda\ox\id)\omega^{AE}, \tau^K\ox \Omega^E\right)^2 \ge (1-\varepsilon)^2,
\end{aligned}
\end{equation}
hence
\[
 \frac{1}{2}\left\|(\Delta\circ\Lambda\ox\id)\omega^{AE} - \tau^K\ox \Omega^E\right\|_1 \leq \sqrt{2\varepsilon}.
\]

Note that from Remark \ref{rem:DIIO}, the map $\Delta\circ\Lambda$ acts on the
incoherent basis states $\proj{x}$ as
%is in particular MIO, so its
%action on the incoherent basis states $\proj{x}$ is
%\[
 % \Lambda(\proj{x}) = \sum_y \Lambda(y|x)\proj{y},
%\]
%with a classical, but potentially noisy, channel $\Lambda(y|x)$. To conclude
%that the distillation map represents a randomness extraction process $G$,
%we would thus have to argue that $\Lambda(y|x)$ is deterministic,
%i.e.~$\Lambda(y|x) = \delta_{y,G(x)}$ for a function $G$.
%This is not a priori clear, but if we knew it, we could then conclude

\begin{align}
\Delta[\Lambda(\op{x}{x})]=\op{G(x)}{G(x)},
\end{align}
which is deterministic. From the achievable distillation IO map, we can construct an extractor and obtain
\begin{equation}
 \ell^{\sqrt{2\varepsilon}}_{\textrm{ext}}(\rho) \geq C^{\frac{\varepsilon}{2}-\eta}_{\min}(\rho) - 2\log\frac{1}{\eta}.
\end{equation}
Recall the result in Theorem \ref{thm:upper},
\begin{equation}
  C^\varepsilon_{d,IO}(\rho) \leq C_{\min}^{\sqrt{\varepsilon(2-\varepsilon)}}(\rho)
                             \leq C_{\min}^{\sqrt{2\varepsilon}}(\rho),
\end{equation}
and we obtain
\begin{equation}
 \ell^{\sqrt{2\varepsilon}}_{\textrm{ext}}(\rho) \geq C^{(\frac{\varepsilon}{2}-\eta)^2/2}_{d,IO}(\rho) - 2\log\frac{1}{\eta},
\end{equation}
finishing the proof.
%\textcolor{red}{AW: I conjecture that the fact that the DIO map $\Lambda$
%distills an approximate maximally coherent state, implies that $\Lambda(y|x)$,
%while perhaps not being exactly deterministic, is close to a deterministic map.
%I'm still battling with the fidelities, but perhaps there is hope...}
\end{proof}

Combining with Eq.~(\ref{extrac-Io}), we have
\begin{equation}
   C^{\sqrt{8\varepsilon}}_{d,IO}(\rho)
	\geq \ell^{\sqrt{2\varepsilon}}_{\textrm{ext}}(\rho)
	\geq C^{(\frac{\varepsilon}{2}-\eta)^2/2}_{d,IO}(\rho) - 2\log\frac{1}{\eta}.
\end{equation}
% and the inclusion of the classes of operation, we have
%\begin{equation}
 % C^{\sqrt{8\varepsilon}}_{d,DIO}(\rho)
  %  \geq C^{\sqrt{8\varepsilon}}_{d,IO}(\rho)
	%\geq \ell^{\sqrt{2\varepsilon}}_{\textrm{ext}}(\rho)
	%\geq C_{d,\mathcal{DIO}}^\varepsilon(\rho).
%\end{equation}

In the regime of vanishingly small $\varepsilon$, the distillable coherence rate
$C_{d,IO}(\rho)$ and $\ell_{\textrm{ext}}(\rho)$ are hence essentially the same. Whether $C_{d,DIO}(\rho)$ and $\ell_{\textrm{ext}}(\rho)$  are the same is still an open problem. Though DIO can commute with dephasing operation, the difficulty stems from that the combination of DIO and extractor (the blue part in Fig.~\ref{fig:extract}(d)) may be not deterministic thus not a valid extractor.

%This we have seen already before by separate
%direct and converse arguments (Sections XX and YY), but the above provides a
%more direct, operational connection.

\section{Discussion}
%Der Worte sind genug gewechselt, lasst mich nun endlich Taten sehn!
%
%The amazing thing is that we managed to get to the end of the paper
%without dying. Let is thank GOD for that and pray, brothers and
%sisters. Everything else is naturally secondary.
%However, reflecting on the results, we may ask the just question,
%who the devil is going to be interested in them. If they are, it
%will become apparent that we did not at all solve all the problems
%that we so loftily defined in the introduction (by now a distant
%memory of the youth of our dear reader); most shamefully, we did
%not determine the SIO distillation cost, neither in the one-shot
%nor the asymptotic setting. For this delusion we ask the sincerest
%forgiveness from the esteemed reader. If any of you buggers insists,
%please feel free to go forth and solve it yourself. Let's see how far
%you get, and if then you're still so disposed to criticize our humble
%efforts from your stinking moral high ground, as it were, you arrogant swine!
%%
%To all others, live long and prosper \cite{spock}.

We have considered the problem of one-shot coherence distillation under
the classes MIO, DIO, IO, and SIO of incoherent operations.
Our results indicate that the distillation rates under
IO, MIO, DIO are roughly the same,
up to different smoothing parameters and universal additive terms.
The results allow us to recover the asymptotic (many-copy) limit, in
which the distillation rates for all these three classes tend to be
relative entropy of coherence, which is consistent with the previous
results in \cite{Winter16,Regula-one-shot,Eric:MIO-CoF}.

The smallest class for which we have been able to show a non-trivial
distillation of coherence is that of dephasing-covariant incoherent IOs,
DIIO = IO$\cap$DIO.
On the other hand, interestingly, there is a gap between distillation rates
under SIO and DIIO, both in the one-shot, and more importantly in the
asymptotic regime. As a matter of fact, we showed that there is
\emph{bound coherence} under SIO; no pure coherence can be distilled under
SIO from these states though they possess coherence, as shown by
distillable coherence under IO.

An interesting future direction is then to study the case involving another
system to help this distillation process, which referred as assisted coherence
distillation for such bound coherence with
SIO \cite{vijayan2018one,regula2018non-asymptotic,PhysRevLett.116.070402,PhysRevX.7.011024}.
Furthermore, our work also connects the distillation of coherence to randomness
extraction. The distillation process is also related to the decoupling in
cryptography. Thus our results also shed light on other quantum information
processing tasks like random number generation, extraction and cryptography.

\medskip
\emph{Note added:} After completion of this work and circulating a
preliminary preprint, Ludovico Lami \emph{et al.}~\cite{lami2018generic} have shown that the bound coherence under SIO is in fact
a generic phenomenon, showing that the fidelity of distilling even a
single cosbit is bounded away from $1$ for all but a measure-zero set of mixed states.

\section*{Acknowledgment}
%Thanks to Kaspar, Melchior and Balthasar for the abundant bling,
%the wonderful weed, and also for the nice body lotion. We thank
%the three magi furthermore for deep discussions on coherence theory,
%nd the holy spirit for his endless supply of physical paradoxes.
%
%This research was made possible due to various generous lashings of
%cash stolen from the taxpayers, and channeled into science by
%the following organisations:
The authors thank Mar{\'\i}a Garc{\'\i}a D{\'\i}az, Xiongfeng Ma, Hongyi Zhou,
Ludovico Lami and Bartosz Regula
for various insightful discussions on coherence theory.

QZ acknowledges support by the National Natural Science Foundation of China Grant No.~11674193.
XY was supported by BP plc and the EPSRC National Quantum Technology Hub in Networked Quantum Information Technology (EP/M013243/1).
EC was supported by the National Science Foundation (NSF) Early CAREER Award No.~1352326.
AW acknowledges support by the Spanish MINECO (project FIS2016-80681-P),
with the support of FEDER funds, and the Generalitat de
Catalunya (CIRIT project 2017-SGR-1127).

\bibliographystyle{IEEEtran}

\vfill\phantom{.}

\end{document}